\documentclass[preprint,authoryear]{elsarticle}\usepackage[a4paper, left=3cm, right=3cm, top=3cm, bottom=3cm]{geometry}
\usepackage{color}
\usepackage{algorithm}
\usepackage{algpseudocode} 
\usepackage[latin1]{inputenc}
\usepackage{xcolor}
\usepackage{subcaption}
\usepackage{ragged2e} 
\usepackage{enumitem}
\usepackage{amsthm}
\usepackage{setspace}
\usepackage{siunitx}
\usepackage{url}
\usepackage{hyperref}
\hypersetup{breaklinks=true}
\usepackage{pifont}
\usepackage{multirow}
\usepackage{scalerel}
\usepackage{amsmath}
\usepackage{tikz}
\usepackage{acmart-taps}
\usepackage{amssymb}  
\usepackage{booktabs}
\usepackage{scalerel} 
\newtheorem{proposition}{Proposition}
\newtheorem{theorem}{Theorem}
\newtheorem{lemma}{Lemma}

\newtheorem{example}{Example}
\newtheorem{corollary}{Corollary}

\theoremstyle{definition}
\newtheorem{definition}{Definition}[section]

\newcommand{\Gr}[1]{{G_I}}

\newcommand\eatpunct[1]{}
\DeclareMathSymbol{\shortminus}{\mathbin}{AMSa}{"39}

\DeclareMathOperator*{\dsum}{\scalerel*{\boxplus}{\sum}}
\DeclareMathAlphabet\mathbfcal{OMS}{cmsy}{b}{n}

 \newcommand{\secminus}{-5pt}

\usepackage{titlesec} 

\makeatletter
\let\OldStatex\Statex
\renewcommand{\Statex}[1][3]{%
  \setlength\@tempdima{\algorithmicindent}%
  \OldStatex\hskip\dimexpr#1\@tempdima\relax}
\makeatother
\begin{document}
\title{Reachability Analysis for Linear Systems with Uncertain Parameters using Polynomial Zonotopes}
\author[1]{Yushen Huang\corref{cor1}\fnref{fn1}}
\ead{yushen.huang@stonybrook.edu}

\author[1]{Ertai Luo\fnref{fn1}}
\ead{erluo@cs.stonybrook.edu}

\author[1]{Stanley Bak}
\ead{sbak@cs.stonybrook.edu}

\author[1]{Yifan Sun}
\ead{yifan.sun@stonybrook.edu}

\fntext[fn1]{These authors contributed equally to this work.}

\cortext[cor1]{Corresponding author}
\address[1]{Stony Brook University, Stony Brook, NY, USA}

\begin{abstract}
In real world applications, uncertain parameters are the rule rather than the exception. 
 We present a reachability algorithm for linear systems with uncertain parameters and inputs using set propagation of polynomial zonotopes.
In contrast to previous methods, our approach is able to tightly capture the non-convexity of the reachable set. 
 Building up on our main result, we show how our reachability algorithm can be extended to handle linear time-varying systems as well as linear systems with time-varying parameters. 
 Moreover, our approach opens up new possibilities for reachability analysis of linear time-invariant systems, nonlinear systems, and hybrid systems. 
 We compare our approach to other state of the art methods, with superior tightness on two benchmarks including a 9-dimensional vehicle platooning system.
Moreover, as part of the journal extension, we investigate through a polynomial zonotope with special structure named multi-affine zonotopes and its optimization problem. 
We provide the corresponding optimization algorithm and experiment over the examples obatined from two benchmark systems, showing the efficiency and scalability comparing to the state of the art method for handling such type of set representation. 
 
\end{abstract}
\maketitle
\vspace{\secminus}
\begin{figure}[ht]
   \centering
   \setlength{\belowcaptionskip}{-13pt}   \includegraphics[width=0.7\columnwidth]{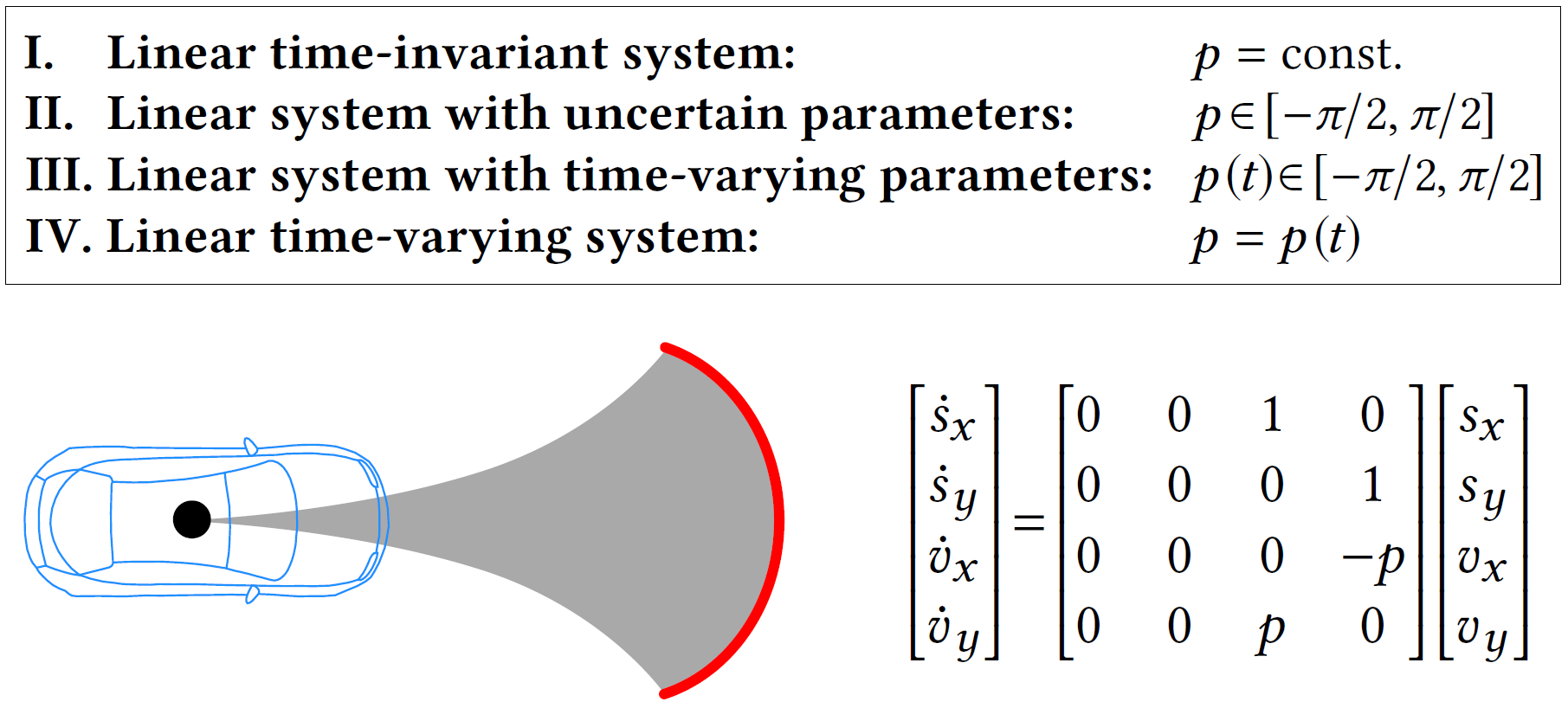}
   \vspace{-6pt}
   \caption{Different classes of linear systems explained on the example of the Dubins car model, where $s_x$, $s_y$ represent the x- and y-position of the car and $v_x$, $v_y$ are the corresponding velocities. The final reachable set is shown in red and the reachable set for the whole time horizon in gray.}
   \label{fig:dubinsCarSchematic}
\end{figure}
\section{Introduction} 
\label{sec:Introduction}
Formal verification is an important aspect of modern cyber-physical systems since many of them, such as autonomous cars, robots collaborating with humans, and power system controllers, are safety-critical. Reachability analysis is one of the most prominent techniques for the formal verification of cyber-physical systems since it can guarantee robust safety in the presence of uncertain parameters and inputs. In particular, if the reachable set of the system does not intersect any unsafe regions, the system is proven to be safe. The focus of reachability analysis is therefore the computation of a tight enclosure of the reachable set, because tightness increases the chances to successfully prove safety and avoid spurious counterexamples. In this paper, we consider linear systems with uncertain parameters and inputs. 
The time-point reachable set for these types of systems is in general non-convex, as it is visualized for the example of the Dubins car in Fig. $\ref{fig:dubinsCarSchematic}$, where the uncertain parameter is the turning rate. While previous methods compute convex enclosures of the reachable set and are consequently limited in their accuracy, we present a new approach that computes very tight non-convex enclosures. 
\journal{Nonlinear Analysis: Hybrid Systems
}
There are many ways to perform reachability analysis for linear time-invariant systems without uncertain parameters (case I. in Fig.~\ref{fig:dubinsCarSchematic}), which is a well studied problem. The reachable set is either computed based on simulations using star sets \citep{Bak2017b,Duggirala2016}, or based on set-propagation using a variety of set representations such as support functions \citep{LeGuernic2010,Frehse2011}, polytopes \citep{Chutinan2003}, ellipsoids \citep{Kurzhanski2000}, and zonotopes \citep{Girard2005}. 
By using Krylov subspaces \citep{Bak2019,Althoff2019} or block-decomposition \citep{Bogomolov2018}, one can compute the reachable set for very high-dimensional systems with thousands of states. While algorithm parameters such as time step size often have to be tuned manually, a recent development in the field is to automate the corresponding tuning process \citep{Wetzlinger2020,Frehse2013}.

Fewer approaches exist for linear systems with uncertain parameters (case II. in Fig.~\ref{fig:dubinsCarSchematic}). 
One approach~\citep{Ghosh2021ReachabilityOL} computes approximations of sets which contain the reachable set with a certain probability. Hereby, candidate sets are first obtained via reachable set calculation for sampled parameter values and then verified using statistical verification techniques. 
A second method~\citep{Lal2015} divides the domain formed by the parameter space and the considered time horizon into a grid in order to construct a piecewise affine approximation of the flow function, where the analytical solution for autonomous linear systems is utilized to obtain the values for the flow function at the grid points. Since the grid size is refined until the piecewise affine approximation achieves a certain accuracy, an enclosure of the reachable set can be obtained by bloating the set defined by the flow function accordingly. 
Moreover, for discrete-time systems it is possible to obtain an enclosure of the reachable set by computing the convex hull of the reachable sets for all vertices of the parameter set in each time step \citep{Silvestre2022}. 
Other methods use set-propagation techniques to compute reachable sets \citep{Ghosh2019,ghosh2021robustness,Althoff2011b}. For discrete-time parametric systems with a special structure, the reachable set can be represented by a generalized star set with bi-linear constraints \citep{Ghosh2019}. The authors extended this approach to continuous-time systems \citep{ghosh2021robustness}, where they bloat the reachable set for the linear system corresponding to the average parameter values. Here, the required bloating is given by an upper bound for the perturbation of the matrix exponential by the parameter uncertainty, which can be determined using sensitivity analysis. 
The method closest to our approach computes tight enclosures of the reachable set for continuous-time parametric systems with uncertain inputs using zonotopes \citep{Althoff2011b}. 
Common ways to represent the parametric uncertainty are interval matrices \citep{ghosh2021robustness,Althoff2011b}, linear matrix equations \citep{Ghosh2021ReachabilityOL,Ghosh2019}, matrix zonotopes \citep{Althoff2011b,Silvestre2022}, and matrix polytopes \citep{Lal2015}, where linear matrix equations and matrix zonotopes are different representations of the same set. 

All approaches listed above assume that the parameters remain constant over time. However, there also exist some methods \citep{Serry2018,Althoff2011a} that consider the more general case of time-varying uncertain parameters (case III. in Fig.~\ref{fig:dubinsCarSchematic}). The framework of differential inequalities can be utilized to compute a hyper-rectangular enclosure of the reachable set efficiently via simulation \citep{Serry2018}. 
Moreover, ellipsoidal enclosures can be obtained via optimization with respect to linear matrix inequality constraints \citep{ZhangZhihao2020}.
Another approach \citep{Althoff2011a} utilizes the Peano-Baker series, which converges to the transition matrix for time-varying systems. Once an enclosure of the transition matrix is obtained, the reachable set can be computed via set-propagation using zonotopes. 

While time-varying linear systems (case IV. in Fig.~\ref{fig:dubinsCarSchematic}) can be solved with methods for time-varying parameters by enclosing the time-varying system matrix with a matrix set in each time step, some approaches specialize on time-varying systems \citep{Serry2022,balandin2020control}. In particular, the first approach \citep{Serry2022} applies a numerical approximation of the transition matrix for set-propagation using zonotopes, where a correction term is added to account for the approximation error. Moreover, for the special case where the initial set is an ellipsoid, the reachable set can be represented by an ellipsoid whose shape-defining matrix is the solution of a linear matrix differential equation \citep{balandin2020control}. 

In this paper, we present an approach for determining the reachable sets of linear systems with uncertain parameters through the use of polynomial zonotopes. The organization of the paper is as follows: Section \ref{sec: Preliminaries} outlines the necessary notations and definitions essential for comprehending the proposed method. Section \ref{sec:Set Operations} examines the fundamental properties of these definitions. Section \ref{sec:ReachabilityAnalysis} describes the algorithm for calculating the reachable set for linear system with uncertain parameters. Section \ref{sec:Application} extends the reachability algorithm to linear systems with time-varying parameters, linear time-varying systems, nonlinear systems and explore potential applications in the filed of hybrid systems. It also provides numerical experiments showing the accuracy and scalability of our reachability algorithms. Section \ref{sec::scalable optimization}, explores an optimization algorithm for multi-affine zonotope, which is a special polynomial zonotope. Leveraging the unique structure of the set, we experimentally shown its efficiency and scalability comparing to state-of-the-art method for handling such type of set representation. Finally, Section \ref{sec:conclusion} concludes the paper and suggests directions for future research.

This paper extends the work presented in the conference version \cite{luo2023reachability} with significant enhancements, detailed as follows:
\begin{itemize}
    \item Identification of additional structural properties in polynomial zonotopes that are beneficial for our algorithm, specifically noticing they are multi-affine zonotopes.
    \item Development of improved optimization algorithms that utilized the identified structural properties for more efficient computation.
    \item Evaluation of the proposed algorithm with two examples to demonstrate its performance improvements and practical applicability.
    \item We provide a detail proof of Lemma \ref{lemma:time}.
\end{itemize}

\vspace{\secminus}

\vspace{2pt}
\vspace{\secminus}
\section{Preliminaries} 
\label{sec: Preliminaries}
\subsection{Notation}
\label{subsec:notation}
We use lowercase letters to represent vectors $v \in \mathbb{R}^{n}$, uppercase letters to represent matrices $A\in \mathbb{R}^{m \times n}$, uppercase calligraphic letter to represent sets $\mathcal{S} \subset \mathbb{R}^{n}$, and uppercase calligraphic bold letters to represent sets of matrices $\bm{\mathcal{A}} \subset \mathbb{R}^{m \times n}$. 
Given a vector $v \in \mathbb{R}^{n}$, $v_{(i)}$ represents the $i$-th entry of the vector. 
For a matrix $A\in \mathbb{R}^{m \times n}$, $A_{(i,j)}$ represents the element at $i$-th row and $j$-th column, $A_{(i,~)}$ denotes the $i$-th row, and $A_{(~,j)}$ denotes the $j$-th column. Given two matrices $A_{1} \in \mathbb{R}^{m\times n_1}$ and $A_{2} \in \mathbb{R}^{m\times n_2}$, their horizontal concatenation is represented by $[A_1~A_2] \in \mathbb{R}^{m\times (n_1 + n_2)}$. The identity matrix of size $n \times n$ is denoted by $I_n$, and $\boldsymbol{0}_{m \times n}$ or $\boldsymbol{1}_{m \times n}$ represents a $m \times n$ matrix consisting of $0$s or $1$s, respectively. 
The infinity norm of a matrix set is defined as $||\bm{\mathcal{A}}||_{\infty} = \max(||A||_{\infty}~|~A\in \bm{\mathcal{A}})$. 
\subsection{Definition}

We first present several definitions and operations which are required throughout the paper. One commonly used set representation for reachability analysis are zonotopes, since they can represent sets very compactly using generator vectors:
\begin{definition}
\label{def:zonotope}
(Zonotope) Given center $c \in \mathbb{R}^n$ and a generator matrix $G \in \mathbb{R}^{n \times h}$, a zonotope is defined as
\begin{equation*}
    \begin{split}
        \mathcal{Z} = \bigg\{ & c+ \sum _{i=1}^{h} \alpha _i\, G_{(~,i)}~ \bigg| ~\alpha_i\in [-1,1] \bigg\}.
    \end{split}
\end{equation*}
We use shorthand notation ~$\mathcal{Z}=\langle c, G\rangle _{Z}$.
\end{definition}
\begin{figure}[ht]
   \centering
   \setlength{\belowcaptionskip}{-10pt}
\includegraphics[width=0.5\columnwidth]{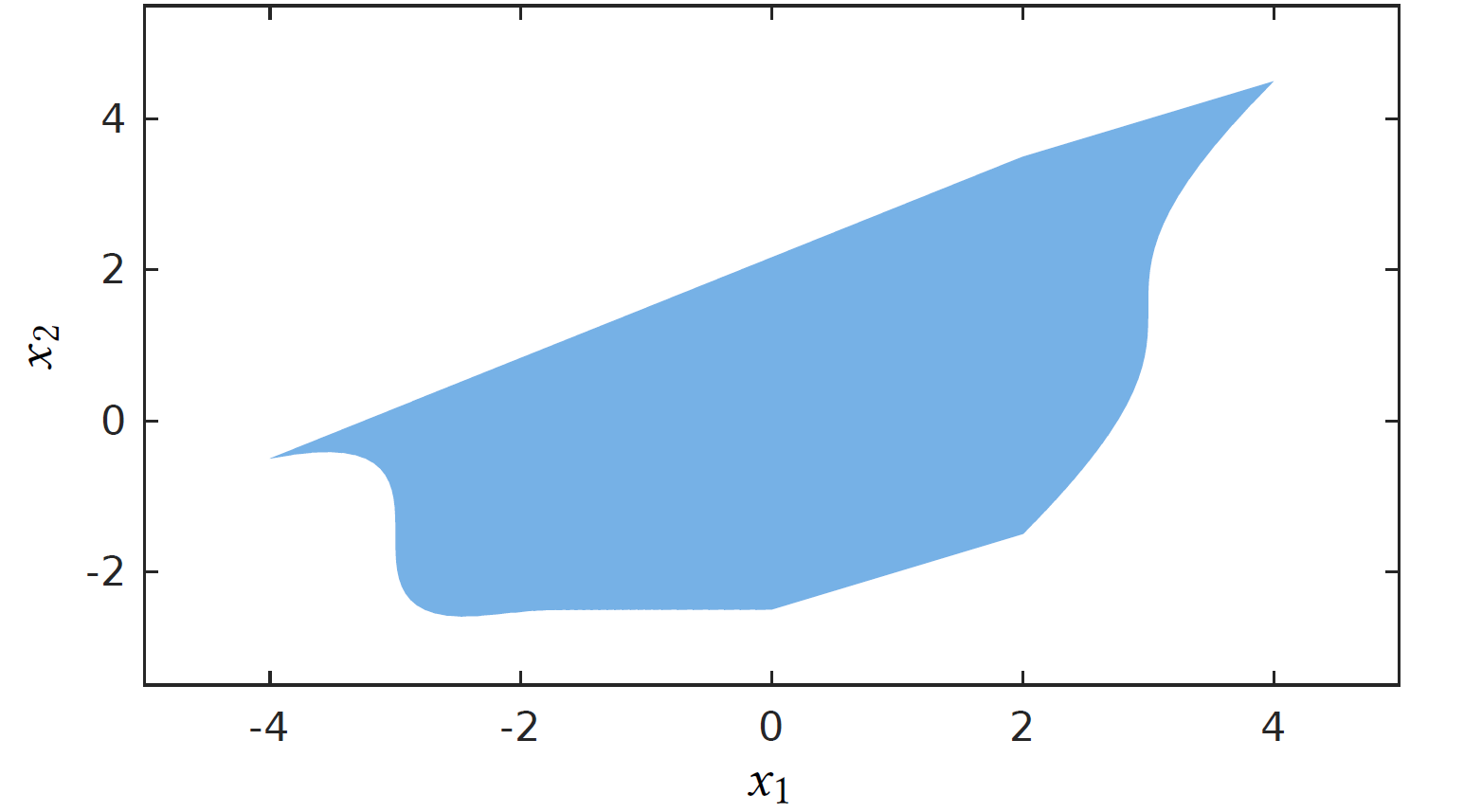}
   \vspace{-18pt}
   \caption{Visualization of the polynomial zonotope in Example \ref{example: single PZ}.}
   \label{fig:Single PZ Filled}
\end{figure}
While zonotopes are always convex, polynomial zonotopes \citep{Althoff2013a} can represent non-convex sets, which allows us to tightly capture the non-convex shape of the reachable set. In this work, we use the sparse representation of polynomial zonotopes \citep{Kochdumper2019a}\footnote{In contrast to \cite[Def.~1]{Kochdumper2019a}, we separate the constant offset $c$ from $G$}:
\begin{definition}
\label{def:polynomialZonotope}
(Polynomial Zonotope) Given a constant offset $c \in \mathbb{R}^n$, a generator matrix of dependent generators $G \in \mathbb{R}^{n \times h}$, a generator matrix of independent generators $G_I \in \mathbb{R}^{n \times q}$, and an exponent matrix $E \in \mathbb{N}_0^{p \times h}$, a polynomial zonotope is defined as
\begin{equation*}
  	\begin{split}
    \mathcal{PZ} \hspace{-2pt}= \hspace{-2pt}\bigg\{ & c+\hspace{-2pt}\sum _{i=1}^h \bigg( \prod _{k=1}^p \alpha _k ^{E_{(k,i)}} \bigg) G_{(~,i)}  + \sum _{j=1}^{q} \beta _j\,G_{I(~,j)}~ \bigg|~\alpha_k, \beta_j \in [-1,1] \bigg\}.
    \end{split}
\end{equation*}
We use the shorthand notation~$\mathcal{PZ}=\langle c, G, G_I, E, id\rangle _{PZ}$, where $id \in \mathbb{N}^{p}$ is a list of non-repeated natural numbers storing a unique identifier for each dependent factor $\alpha_k$.
\end{definition}
Let us demonstrate polynomial zonotopes by an example:
\begin{example}
\label{example: single PZ}
The polynomial zonotope with sparse representation
\begin{equation*}
    \mathcal{PZ} = \bigg\langle 
    \begin{bmatrix} 0 \\ 0 \end{bmatrix},
    \begin{bmatrix} 2 & 0 & 1 \\ 1 & 2 & 1\end{bmatrix},
    \begin{bmatrix} 1 \\ 0.5\end{bmatrix},
    \begin{bmatrix} 1 & 0 & 1 \\ 0 & 1 & 3\end{bmatrix},
    \begin{bmatrix} 1\\ 2 \end{bmatrix}
    \bigg\rangle_{PZ}
\end{equation*}
defines the non-convex set:
\begin{equation*}
    \mathcal{PZ}\hspace{-2pt}= \hspace{-2pt}\bigg\{\hspace{-2pt}
    \begin{bmatrix} 0 \\ 0 \end{bmatrix} + \alpha_1 \hspace{-1.5pt}\begin{bmatrix} 2 \\ 1 \end{bmatrix} + 
    \alpha_2 \hspace{-1.5pt} \begin{bmatrix} 0 \\ 2 \end{bmatrix} + 
    \alpha_1\alpha_2^{3} \hspace{-1.5pt} \begin{bmatrix} 1 \\ 1\end{bmatrix} +
    \beta_1 \hspace{-1.5pt} \begin{bmatrix} 1 \\ 0.5 \end{bmatrix}\hspace{0.5pt}\bigg |~\alpha_1, \alpha_2, \beta_1 \hspace{-2pt}\in\hspace{-2pt} [-1, 1]
    \bigg\},
\end{equation*}
which is visualized in Fig.$~\ref{fig:Single PZ Filled}$. 


\end{example}
The parameter uncertainty for linear systems can be represented using sets of matrices. In this work we use the two different matrix set representations interval matrices and matrix zonotopes.
\begin{definition}
\label{def:interval matrix}
(Interval Matrix) Given a pair of lower and upper matrices $\underline{A}, \overline{A} \in \mathbb{R}^{m \times n}$, an interval matrix is defined as
\begin{equation*}
    \begin{split}
         \bm{\mathcal{A}} = \begin{bmatrix} [\underline{A}_{(1,1)}, \overline{A}_{(1,1)}]& \ldots &[\underline{A}_{(1,n)}, \overline{A}_{(1,n)}] \\ \vdots & \ddots & \vdots\\ [\underline{A}_{(m,1)}, \overline{A}_{(m,1)}]& \ldots &[\underline{A}_{(m,n)}, \overline{A}_{(m,n)}]\end{bmatrix}.
    \end{split}
\end{equation*}
We use the shorthand notation $\bm{\mathcal{A}} = \langle\underline{A}, \overline{A}\rangle_{IM}$.
\end{definition}
Similar to how zonotopes generalize intervals, matrix zonotopes can be used to generalize interval matrices.
\begin{definition}
(Matrix Zonotope) Given center matrix $A^{(0)} \in \mathbb{R}^{m \times n}$ and generator matrices $A^{(1)}, \ldots,A^{(w)} \in \mathbb{R}^{m \times n}$, a matrix zonotope is defined as
\label{def:matrixZonotope}
    \begin{equation*}
    \bm{\mathcal{A}} = \bigg\{ A^{(0)} + \sum _{l=1}^w \rho_l\, A^{(l)}~ \bigg| \ \rho_l \in [-1,1] \bigg\}.
    \end{equation*}
    We use shorthand notation $\bm{\mathcal{A}} = \langle A^{(0)}, A^{(1)}, \ldots,A^{(w)},id\rangle_{MZ}$
    with $id \in \mathbb{N}^w$ storing the unique identifiers for each factor $\rho_l$.
\end{definition}

The identifiers allow us to preserve dependencies between sets. To achieve that, we require the operation $\textproc{mergeID}$ \cite[Prop.~1]{Kochdumper2019a} to bring the exponent matrices and id lists into a common format. 
\begin{definition}
Given two exponent matrices $E_1 \in \mathbb{R}^{p_1 \times h_1},$ $E_2 \in \mathbb{R}^{p_2 \times h_2}$ and two identifier lists $id_1 \in \mathbb{R}^{p_1}$, $id_2 \in \mathbb{R}^{p_2}$, {\normalfont \textproc{mergeID}} returns two aligned matrices $\overline{E}_1,~ \overline{E}_2$ and a list of adjusted identifiers $\overline{id}$:
\begin{equation*}
    \overline{E}_1,~ \overline{E}_2,~\overline{id} \gets {\normalfont \textproc{mergeID}}(id_1, id_2, E_1, E_2)
\end{equation*}
with
\begin{equation*}
    \begin{split}
        &\mathcal{K} = \big \{i ~|~id_{2(i)} \not\in id_1 \big\} := \{i_1,~\ldots,~ i_{k}\},\\
        &\overline{id} = \big[id_1^{T} \hspace{5pt} id_{2(i_1)} \ldots ~id_{2(i_k)} \big]^{T},\hspace{5pt} \overline{E}_1 = 
        \begin{bmatrix}E_1\\\boldsymbol{0}_{k \times h_1}\end{bmatrix} \in \mathbb{R}^{(p_1 + k) \times h_1},\\
        &\overline{E}_{2(i, ~)} = \begin{cases} E_{2(j,~)}, & \text{if}\hspace{5pt}\exists j~\,\overline{id}_{(i)} = id_{2(j)}\\ \boldsymbol{0}_{1 \times h_2}, & \text{otherwise} \end{cases},~i = 1, \ldots, p_1 + k.
    \end{split}
\end{equation*}
\end{definition}
Sometimes we also need to destroy dependencies to guarantee the correctness of our reachability algorithm. 
For this we introduce the operation $\textproc{uniqueID}(p)$, which returns a vector of unique and non-existing identifiers with length $p$. 
Moreover, we define a convenience operation $\textproc{fresh}(\mathcal{PZ})$ $= \langle c, G, G_{I}, E,$ $\textproc{uniqueID}(p)\rangle_{PZ}$  replacing all identifiers of $\mathcal{PZ}$. The $\textproc{fresh}$ operator can equally be applied to matrix zonotopes. For the ease of notation we denote by $\textproc{fresh}(\mathcal{PZ},\bm{\mathcal{A}})$ that all dependent factors of $\mathcal{PZ}$ that are not part of the matrix zonotope $\bm{\mathcal{A}}$ are equipped with new unique identifiers. The operation $\textproc{eval}(\mathcal{PZ}, id, val)$ replaces each dependent factor of $\mathcal{PZ}$ that has the identifier $\textit{id}$ with a constant value $val \in [-1, 1]$. The operation also applies to matrix zonotopes. As an example, performing $\textproc{eval}(\mathcal{PZ}, 2, 0.5)$ on the polynomial zonotope from Example $\ref{example: single PZ}$ yields the set
\begin{equation*}
    \textproc{eval}(\mathcal{PZ}, 2, 0.5) = \\
    \bigg\{
    \alpha_1\begin{bmatrix} 2 \\ 1 \end{bmatrix} + 
    0.5\begin{bmatrix} 0 \\ 2 \end{bmatrix} + 
    \alpha_1(0.5)^{3}\begin{bmatrix} 1 \\ 1\end{bmatrix} +
    \beta_1\begin{bmatrix} 1 \\ 0.5 \end{bmatrix}~\bigg |~
    \alpha_1, \beta_1 \in [-1, 1]\bigg\}.\\
\end{equation*}
In this paper we use set operations to perform linear maps with either a matrix or a matrix set, as well as sums that can preserve or drop dependencies. 
Given sets $\mathcal{S}_1, \mathcal{S}_2 \subset \mathbb{R}^n$, a numerical matrix $A \in \mathbb{R}^{m\times n}$ and a matrix set $\bm{\mathcal{A}} \subset \mathbb{R}^{m \times n}$, these operations are:
\begin{align}
    A\,\mathcal{S}_1 & = \{A\,s~|~s \in \mathcal{S}_1\}, \label{def:numerical matrix multiply a set of states} \\
    \bm{\mathcal{A}}~\mathcal{S}_1 & = \{ A\,s~|~A\in \bm{\mathcal{A}}, s \in \mathcal{S}_1\}, \label{def:matrix set multiply a set of states} \\
    \mathcal{S}_1 \oplus \mathcal{S}_2 & = \{s_1 + s_2~|~s_1 \in \mathcal{S}_1, s_2\in\mathcal{S}_2\}. \label{def:Minkowski sum between two set of states}
\end{align}
Given two polynomial zonotopes $\mathcal{PZ}_1=\langle c_1, G_1, G_{I1}, E_1, id_1\rangle_{PZ}$, $\mathcal{PZ}_2=\langle c_2, G_2, G_{I2}, E_2, id_2\rangle_{PZ}$ $\subset \mathbb{R}^n$ and a zonotope $\mathcal{Z} = \langle c_z, G_z \rangle_Z$ $\subset \mathbb{R}^n$, the Minkowski sum can be computed as \cite[Prop.~9]{Kochdumper2019a}:
\begin{eqnarray*}
    \label{op:minkowski sum with PZ}
        \mathcal{PZ}_1 \oplus \mathcal{PZ}_2 &=&  \hspace{1pt} \bigg \langle c_1+c_2,[G_1~ G_2],[G_{I1}~ G_{I2}],\;
        \begin{bmatrix} E_1 & \mathbf{0} \\ \mathbf{0} & E_2 \end{bmatrix},\textproc{uniqueID}(id_1,id_2)~\bigg \rangle_{PZ},\\
    \label{op:minkowski sum with Z}
    \mathcal{PZ}_1 \oplus \mathcal{Z} &=& \langle c_1+c_z, G_1,[G_{I1}~G_z], E_1, id_1 \rangle_{PZ}.
\end{eqnarray*}
While Minkowski sum drops dependencies between sets, the exact sum operation $\boxplus$ explicitly keeps these dependencies \cite[Prop. 10]{Kochdumper2019a}:
\begin{equation}
    \label{op:exact sum with PZ}
    \begin{split}
        &\mathcal{PZ}_1~\boxplus~\mathcal{PZ}_2 = \big\langle c_1+c_2, [G_1~ G_2], [G_{I1}~G_{I2}], [\overline{E}_1~\overline{E}_2],\overline{id} \big \rangle_{PZ},
    \end{split}
\end{equation}
\begin{figure}[ht]
   \centering
   \setlength{\belowcaptionskip}{-12pt}
   \hspace{-10pt}
   \includegraphics[width=0.5\columnwidth]{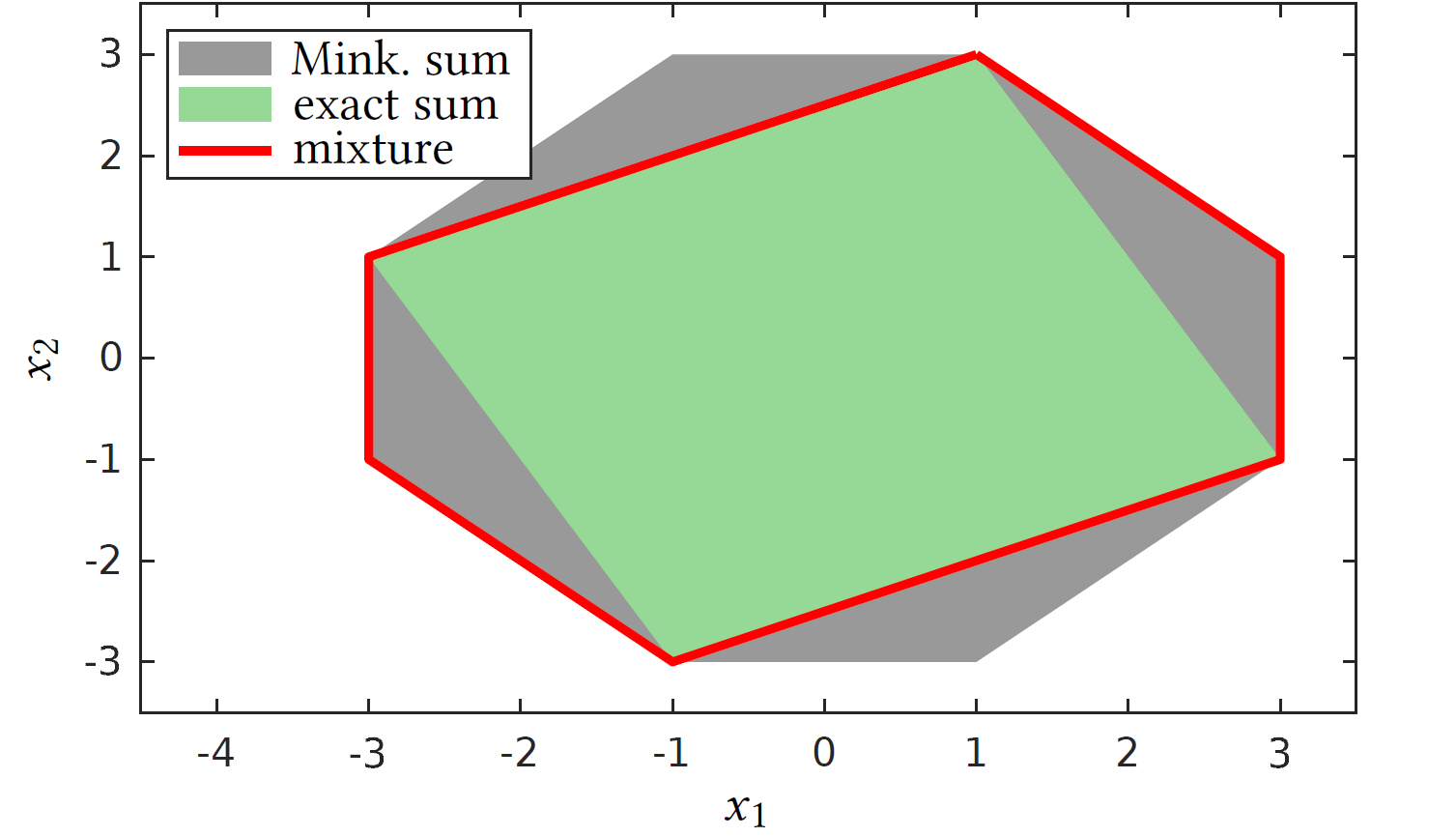}
   \vspace{-19pt}
   \caption{Comparison from Example~\ref{example: compare minkowski sum and direct sum} between Minkowski sum $A~\mathcal{PZ} \oplus \mathcal{PZ}$, exact sum $A~ \mathcal{PZ} \boxplus \mathcal{PZ}$, and a mixture of both obtained by destroying the dependencies for the second dependent factor.}
   \label{fig:Minkowski and Exact Sum}
\end{figure}
where the exponent matrices and identifier are first merged using $\overline{E}_1,~ \overline{E}_2,~\overline{id} \gets \textproc{mergeID}(id_1, id_2, E_1, E_2)$. 
The operation $\textproc{compact}$ \cite[Prop.~2]{Kochdumper2019a} is applied after each exact sum to remove redundancies such as equal columns in the exponent matrix. 
%
Let us demonstrate the difference between Minkowski sum and exact sum:
\begin{example}
\label{example: compare minkowski sum and direct sum}
We consider the polynomial zonotope $\mathcal{PZ} = \langle \mathbf{0}_{2 \times 1},I_n, \linebreak[3] [~],I_n,[1 \hspace{5pt} 2]^T\rangle_{PZ}$ together with the numerical matrix $A = [[1 \hspace{5pt} 1]^{T} \hspace{5pt} \linebreak[3] [-1 \hspace{5pt} 1]^{T}]$. The comparison of the Minkowski sum $A~\mathcal{PZ} \oplus \mathcal{PZ}$ with the exact sum $A~ \mathcal{PZ} \boxplus \mathcal{PZ}$ visualized in Fig. $\ref{fig:Minkowski and Exact Sum}$ demonstrates that the Minkowski sum yields an over-approximation of the exact result since the dependencies between the sets $A~\mathcal{PZ}$ and $\mathcal{PZ}$ are destroyed. Moreover, if we only destroy the dependencies for some factors, as for example done by the operation ${\normalfont \textproc{fresh}}(\mathcal{PZ},\bm{\mathcal{A}})$, we obtain a mixture between Minkowski sum and exact sum. 
\end{example}

We additionally require several other operations that are provided in a previous paper on polynomial zonotopes \cite{Kochdumper2019a}: The operator $\textproc{zonotope}(\mathcal{PZ})$ \cite[Prop.~ 5]{Kochdumper2019a} returns a tight enclosing zonotope of a polynomial zonotope. Moreover, $~\textproc{reduce}(\mathcal{PZ}, \rho)$ \cite[Prop.~ 16]{Kochdumper2019a} reduces the representation size and returns a polynomial zonotope with order smaller or equal to $\rho$ which encloses $\mathcal{PZ}$. 

We now introduce a special kind of polynomial zonotope
\begin{definition}
A \emph{multi-affine zonotope} is a  polynomial zonotope in which the exponent matrix $E$ exclusively composed of $0$ or $1$ elements.
The multi-affine zonotope has similar shorthand notation ~$\mathcal{MAZ}=\langle c, G, G_I, E, id\rangle _{MAZ}$.
\label{def:multiaffine-zonotope}
\end{definition}

\vspace{2pt}
\vspace{\secminus}
\section{SET OPERATIONS}
\label{sec:Set Operations}
For our proposed reachability algorithm,  we require additional operations on polynomial zonotopes, which we derive in this section.
We begin with the multiplication of a matrix zonotope with a polynomial zonotope, for which we consider the case without independent generators first and present the general case afterwards.
\begin{proposition}
(Matrix Zonotope Multiplication) Given a matrix zonotope $\bm{\mathcal{A}} =\langle A^{(0)}, A^{(1)}, \ldots,A^{(w)},id_{mz}\rangle_{MZ} \subset \mathbb{R}^{m \times n}$ and a polynomial zonotope $\mathcal{PZ} = \langle c, G, [~], E, id_{pz}\rangle_{PZ} \subset \mathbb{R}^n$, their multiplication is
\label{prop:mtimes}
\begin{equation*}
\label{cal:mtimes cal}
\begin{split}
    \bm{\mathcal{A}}~\mathcal{PZ}= \hspace{1pt} &\big\langle A^{(0)}c, [A^{(0)}G \hspace{5pt} A^{(1)}c ~\ldots ~ A^{(w)}c \hspace{5pt} A^{(1)}G ~ \ldots ~ A^{(w)}G],\\
    & \hspace{95pt} [~],[\widehat{E}_1 \hspace{5pt} \widehat{E}_2 \hspace{5pt} \overline{E}_1 ~ \ldots ~ \overline{E}_w], \widehat{id} \big\rangle_{PZ}
\end{split}
\end{equation*}
with
\begin{equation}
    \label{cal:mtimes new exponent matrices}
    \overline{E}_l = \widehat{E}_{2(~,l)}\cdot\boldsymbol{1}_{1 \times h}+ \widehat{E}_{1}, \quad  l =1,\ldots, w,
\end{equation}
where $\widehat{E}_1, \widehat{E}_2, \widehat{id} \gets \textproc{mergeID}(id_{pz}, id_{mz},E,I_w)$ brings the exponent matrices into a common format.
\end{proposition}
\begin{proof}
Inserting the definitions of matrix zonotopes and polynomial zonotopes in Def. $\ref{def:polynomialZonotope}$ and Def. $\ref{def:matrixZonotope}$ into the definition of the multiplication with a matrix set in \eqref{def:matrix set multiply a set of states} yields:
\allowdisplaybreaks
\begin{align*}
    &\bm{\mathcal{A}} ~ \mathcal{PZ} \overset{\eqref{def:matrix set multiply a set of states}}{=} \big\{ As \ \big|\ A \in \bm{\mathcal{A}}, s \in \mathcal{PZ} \big \}\\[5pt]
    & = \bigg\{ \bigg(A^{(0)} \hspace{-3pt}+\hspace{-3pt} \sum _{l=1}^w \rho_l A^{(l)}\bigg)\bigg( c\hspace{-2pt}+\hspace{-3pt} \sum _{i=1}^h \bigg( \prod _{k=1}^m \alpha _k ^{E_{(k,i)}} \bigg) G_{(~,i)} \bigg)~\bigg|~\rho_l, \alpha_k \in [-1,1] \bigg\}\\[5pt]
    &= \bigg\{ A^{(0)}c + \sum _{i=1}^h \bigg( \prod _{k=1}^m \alpha _k ^{E_{(k,i)}} \bigg) A^{(0)} G_{(~,i)} + \sum _{l=1}^w \rho_l A^{(l)}c \hspace{1pt}+\\
    &\hspace{20pt}\sum _{l=1}^w\sum _{i=1}^h \bigg( \prod _{k=1}^m \alpha _k ^{E_{(k,i)}} \bigg) \rho_l  A^{(l)}G_{(~,i)}~ \bigg|~\rho_l, \alpha_k, \in [-1,1] \bigg\} \\[10pt]
        & =\big \langle A^{(0)}c, [A^{(0)}G \hspace{5pt} A^{(1)}c ~ \ldots ~ A^{(w)}c \hspace{5pt} A^{(1)}G ~ \ldots ~ A^{(w)}G],[~],\\
        &\hspace{135pt}[\widehat{E}_1 \hspace{5pt} \widehat{E}_2 \hspace{5pt} \overline{E}_1 ~ \ldots ~ \overline{E}_w], \widehat{id} \big\rangle_{PZ}.
\end{align*}
Here, we exploit that any matrix zonotope has an exponent matrix $I_w$ because each generator matrix is only multiplied with a single factor. Since the operation $\textproc{mergeID}$ brings the exponent matrices to a common representation, the multiplication of the factors $\big( \prod _{k=1}^m \alpha _k ^{\substack{E_{(k,i)} \\ \vspace{-5pt}}}\big)~\rho_l$ can be considered by an addition of the $l$-th column of exponent matrix $\widehat{E}_2$ to $i$-th column of $\widehat{E}_1$ as done in $\eqref{cal:mtimes new exponent matrices}$.
\end{proof}
Based on Prop. $\ref{prop:mtimes}$, the multiplication with powers of matrix sets can be computed as follows:
\begin{corollary}
\label{corollary: higher order multiplication}
(Higher Order Multiplication) Given a matrix zonotope $\bm{\mathcal{A}}$ and a polynomial zonotope $\mathcal{PZ}$, the set 
\begin{equation*}
    \big\{A^k s ~\big |~ A\in \bm{\mathcal{A}}, s \in \mathcal{PZ} \big \} = \bm{\mathcal{A}}^{k}~\mathcal{PZ}
\end{equation*} 
can be computed using $k$ multiplications $\bm{\mathcal{A}} \, \big(~\ldots ~ \big(\bm{\mathcal{A}}~(\bm{\mathcal{A}}~\mathcal{PZ})\big)\big)$.
\end{corollary}
Any polynomial zonotope can equivalently be represented as a polynomial zonotope without independent generators by introducing new identifiers for the independent factors and expanding the exponent matrix \cite[Prop. 1]{Kochdumper2020a}. Therefore, also for the case with independent generators the multiplication with a matrix zonotope can always be performed exactly. In practice, for computational efficiency and based on the assumption that the independent generators are relatively small, we compute a tight enclosure of the multiplication with the independent generators: 
\begin{proposition}
(Matrix Zonotope Multiplication) Given a matrix zonotope $\bm{\mathcal{A}} =\langle A^{(0)}, A^{(1)}, \ldots,A^{(w)},id_{mz}\rangle_{MZ} \subset \mathbb{R}^{m \times n}$ and a polynomial zonotope $\mathcal{PZ}$ = $\langle c, G, G_I, E, id_{pz}\rangle_{PZ} \subset \mathbb{R}^n$, their multiplication can be tightly enclosed by
\begin{equation*}
\begin{split}
    \bm{\mathcal{A}}~\mathcal{PZ} \subseteq \big \langle \widehat{c},\widehat{G},[A^{(0)}G_I ~ \dots ~ A^{(w)}G_I], \widehat{E},\widehat{id} \big \rangle_{PZ},
\end{split}
\end{equation*}
where 
\begin{equation*}
    \big \langle \widehat{c},\widehat{G}, [~], \widehat{E},\widehat{id} \big \rangle_{PZ} = \bm{\mathcal{A}}~ \big \langle c, G, [~], E, id_{pz} \big \rangle_{PZ}
\end{equation*}
is computed using Prop. $\ref{prop:mtimes}$.
\end{proposition}
\begin{proof}
The independent generators of $\mathcal{PZ}$ can be viewed as a zero-centered zonotope,
so $\mathcal{PZ} = \mathcal{PZ}_D \oplus \mathcal{Z}_I$ is given by the Minkowski sum between the dependent part $\mathcal{PZ}_D$ = $\langle c, G, [~], E, \\id_{pz}\rangle_{PZ}$ and the independent part $\mathcal{Z}_I = \langle \boldsymbol{0}_{n \times 1}, G_I \rangle_{Z}$. Therefore, the multiplication can be over-approximated by:
\begin{equation*}
 \bm{\mathcal{A}}~ \mathcal{PZ}\subseteq(\bm{\mathcal{A}}~ \mathcal{PZ}_D) \oplus (\bm{\mathcal{A}}~ \mathcal{Z}_I),   
\end{equation*}
where the result for the multiplication between a matrix zonotope with a zonotope is given by \cite[Prop.~4]{Althoff2011a}.
\end{proof}
The core operation for reachability analysis is the linear map with the matrix exponential. In particular, we require the multiplication with the matrix exponential of two matrix sets $\bm{\mathcal{A}}$ and $\bm{\mathcal{B}}$:
\begin{proposition}
(Multiplication with Matrix Exponential) Given two matrix zonotopes $\bm{\mathcal{A}}$, $\bm{\mathcal{B}} \subset \mathbb{R}^{n \times n}$, a polynomial zonotope $\mathcal{PZ} \subset \mathbb{R}^{n}$, and a Taylor order $\kappa$ satisfying $\epsilon=\frac{||\bm{\mathcal{A}}||_{\infty}||\bm{\mathcal{B}}||_{\infty}}{\kappa+2}<1$, the set    $e^{\bm{\mathcal{A}}\bm{\mathcal{B}}}\mathcal{PZ}$ can be tightly enclosed by
\label{prop:raise_to_e}
\begin{equation*}
    \begin{split}
    \label{op:over-approx exp}
        e^{\bm{\mathcal{A}}\bm{\mathcal{B}}}~\mathcal{PZ}\subseteq \bigg(\dsum_{i=0}^{\kappa}\frac{\bm{\mathcal{A}}^i\bm{\mathcal{B}}^i}{i!}~\mathcal{PZ}\bigg)\oplus \Big(\bm{\mathcal{E}}~ {\normalfont \textproc{zonotope}}(\mathcal{PZ})\Big),
    \end{split}
\end{equation*}
where the interval matrix
\begin{equation*}
    \begin{split}
        \label{def:taylorremainder}
        &\bm{\mathcal{E}} = \langle-\boldsymbol{1}_{n \times n}, \boldsymbol{1}_{n \times n}\rangle_{IM}\frac{(||\bm{\mathcal{A}}||_{\infty}||\bm{\mathcal{B}}||_{\infty})^{\kappa+1}}{(\kappa+1)!}\frac{1}{1-\epsilon}
    \end{split}
\end{equation*}
encloses the remaining terms of the truncated Taylor series. The sets $\frac{\bm{\mathcal{A}}^i\bm{\mathcal{B}}^i}{i!}~\mathcal{PZ}$ are calculated according to Corollary $\ref{corollary: higher order multiplication}$, 
and the multiplication of an interval matrix with a zonotope using \cite[Thm.~4]{Althoff2007c}.
\end{proposition}
\begin{proof}
The matrix exponential is defined by its Taylor series:
\begin{equation*}
    \begin{split}
        & e^{\bm{\mathcal{A}}\bm{\mathcal{B}}}~\mathcal{PZ}
        =\bigg\{\bigg (\sum_{i=0}^{\infty}\frac{A^iB^i}{i!} \bigg )\,s~\bigg|~A \in \bm{\mathcal{A}}, B \in \bm{\mathcal{B}}, s\in \mathcal{PZ}\bigg\} \overset{\text{\cite[Thm.~3.2]{Althoff2010a}}}{=}\\[5pt]
        &\bigg\{\bigg(\sum_{i=0}^{\kappa}\frac{A^iB^i}{i!}\bigg)\,s + \underbrace{\bigg( \sum_{i=\kappa+1}^{\infty}\frac{A^iB^i}{i!}\bigg)}_{\subseteq \bm{\mathcal{E}}}s~\bigg|~A \in \bm{\mathcal{A}}, B \in \bm{\mathcal{B}}, s\in \mathcal{PZ}\bigg\} \subseteq \\
        &\bigg\{\bigg(\sum_{i=0}^{\kappa}\frac{A^iB^i}{i!}\bigg)\,s~\bigg|~A \in\hspace{-2pt} \bm{\mathcal{A}}, B \in\hspace{-2pt} \bm{\mathcal{B}}, s \in\hspace{-2pt} \mathcal{PZ}\bigg\} \oplus \big\{E\,s~\big|~E \in \bm{\mathcal{E}},~ s\in \mathcal{PZ} \big\}\\[5pt]
        &\subseteq \bigg(\dsum_{i=0}^{\kappa}\frac{\bm{\mathcal{A}}^{i}\bm{\mathcal{B}}^i}{i!}\mathcal{PZ}\bigg)\oplus \big(\bm{\mathcal{E}}~ \textproc{zonotope}(\mathcal{PZ})\big).
    \end{split}
\end{equation*}
Since the remainder interval matrix $\bm{\mathcal{E}}$ is expected to be small for a sufficiently large Taylor order, we add the corresponding vectors to the independent generators by applying a zonotope enclosure. 
\end{proof}

\vspace{\secminus}
\section{Reachability Analysis}
\label{sec:ReachabilityAnalysis}
We now explain how to compute reachable sets for linear systems with uncertain parameters (case II. in Fig.~\ref{fig:dubinsCarSchematic}) using the operations on polynomial zonotopes derived in the previous section.

\vspace{\secminus}
\subsection{Problem Definition}
We consider a linear system with uncertain parameters
\begin{equation}
    \label{def:probdef}
    \dot x(t) = A \, x(t) + B \, u(t),\quad A \in \bm{\mathcal{A}}, B \in \bm{\mathcal{B}}~,
\end{equation}
where $\bm{\mathcal{A}} \subset \mathbb{R}^{n\times n}$, $\bm{\mathcal{B}}\subset\mathbb{R}^{n\times m}$ are matrix zonotopes. The initial state $x(0)$ is uncertain within $\mathcal{X}_0 \subset \mathbb{R}^n$, and the input $u(t)$ is uncertain within $\mathcal{U} \subset \mathbb{R}^{m}$. The reachable set is defined as follows: 
\begin{definition}
(Reachable Set) Given an initial set $\mathcal{X}_0$, an input set $\mathcal{U}$, and matrix zonotopes $\bm{\mathcal{A}}$ and $\bm{\mathcal{B}}$, the reachable set for the system in \eqref{def:probdef} at time $t \ge 0$ is
\begin{equation*}
    \mathcal{R}(t) = \big \{\xi(A,B,t, x(0),u(\cdot))~ \big |~ 
    A\in \bm{\mathcal{A}},~ B \in \bm{\mathcal{B}}, ~ x(0)\in \mathcal{X}_0,
    \forall s\in[0, t]: ~ u(\hspace{1pt}s\hspace{1pt}) \in \mathcal{U} \big \},
\end{equation*}
where $\xi(A, B, t, x(0),u(\cdot))$ denotes the solution of \eqref{def:probdef} at time t for numerical matrices $A$ and $B$, initial state $x(0)$, and input signal $u(\cdot)$.
\end{definition}
The goal is to calculate the reachable set for a finite time interval $[0, t_{end}]$, where time starts at $0$ without loss of generality. Since the exact reachable set cannot be computed in general, we instead aim to calculate a tight enclosure of the reachable set.
As commonly done in reachability analysis, we compute the reachable set for consecutive time intervals $\tau_k = [k\Delta t, (k+1)\Delta t]$ with time step $\Delta t$, so that the reachable set for the whole time horizon is given as $\mathcal{R}([0,t_{end}]) = \bigcup_{k=0}^{(t_{end} / \Delta t) - 1} \mathcal{R}(\tau_k)$. 
According to the superposition principle, the analytical solution to $\eqref{def:probdef}$ for numerical matrices $A$ and $B$ is given by the addition of the homogeneous solution for the initial state and the particular solution due to inputs:
\begin{equation}
    \label{def:singlepointsol}
    \xi(A, B, t, x(0),u(\cdot)) = e^{At}x(0)+ \int_{0}^te^{A(t-s)}B \, u(s) \, ds.
\end{equation}
This concept also carries over to the case with uncertain parameters, so that we can first compute the homogeneous solution $\mathcal{H}(\tau_k)$ and the particular solution $\mathcal{P}(\tau_k)$ separately, and later combine them to obtain the reachable set $\mathcal{R}(\tau_k)$.

\vspace{\secminus}
\subsection{Homogeneous Solution}
\label{section:homogeneous solution}
First, we consider the homogeneous solution $\mathcal{H}(\tau_0)$ in the first time interval $\tau_0 = [0, \Delta t]$, which can be calculated as
\begin{equation*}
    \label{alg:homogeneous sol}
    \begin{split}
        \mathcal{H}(\tau_0)=\big\{e^{At}x(0)~\big|~A\in \bm{\mathcal{A}}, t \in \tau_0, x(0) \in \mathcal{X}_{0}\big\}=e^{\bm{\mathcal{A}}\bm{\mathcal{T}}}\, \mathcal{X}_0,
    \end{split}
\end{equation*}
where we represent time interval $\tau_0 =  [0, \Delta t]$ by a matrix zonotope
\begin{equation}
\label{def: time matrix zonotope}
\begin{split}
&\bm{\mathcal{T}} = [0, \Delta t] \cdot I_n= \big \langle T, T, \textproc{uniqueID}(1) \big \rangle_{MZ},\quad T = 0.5 \, \Delta t \, I_n.
\end{split}
\end{equation}
We can compute a tight enclosure of the set $e^{\bm{\mathcal{A}}\bm{\mathcal{T}}}\, \mathcal{X}_0$ using Prop. $\ref{prop:raise_to_e}$. Given the homogeneous solution for the first time interval, the solutions for all remaining time intervals can be computed by propagating the initial solution forward in time via multiplication with the matrix exponential $e^{\bm{\mathcal{A}}\Delta t}:$
\begin{equation*}
    \mathcal{H}(\tau_k) = \underbrace{e^{\bm{\mathcal{A}}\Delta t}~\ldots~ e^{\bm{\mathcal{A}}\Delta t}}_{k}~ e^{\bm{\mathcal{A}}\bm{\mathcal{T}}}~ \mathcal{X}_0.
\end{equation*}
\vspace{-3pt}
Note that even if $\mathcal{X}_0$ is convex, $\mathcal{H}(\tau_k)$ is generally non-convex.

\vspace{\secminus}
\subsection{Particular Solution}
\label{section:partiuclar solution}
According to \cite[Appendix~A]{Althoff2010a}, the particular solution at time $t$ for numerical matrices $A\in \bm{\mathcal{A}}$ and $B \in \bm{\mathcal{B}}$ can be enclosed by
\begin{equation}
    \label{def: input solution set}
       \mathcal{P}_{A,B}(t) 
       = \bigg\{\int_{0}^te^{A(t-s)}~B\,u(t)\,ds~\bigg|~ u(t) \in \mathcal{U}\bigg\}
       \subseteq \bigoplus_{i=0}^{\kappa}\bigg(\frac{A^{i}t^{i+1}}{(i+1)!}~B\,\mathcal{U}\bigg)\oplus \big (t \, \bm{\mathcal{E}}_t \,  \textproc{zonotope}(B\,\mathcal{U})\big ),
\end{equation}
where 
\begin{equation}
    \label{def: remainder at time t}
    \bm{\mathcal{E}}_t = \langle-\boldsymbol{1}_{n \times n}, \boldsymbol{1}_{n \times n}\rangle_{IM}\frac{(||\bm{\mathcal{A}}||_{\infty} t)^{\kappa+1}}{(\kappa+1)!}\frac{1}{1-\epsilon},~ \epsilon=\frac{||A||_{\infty} t}{\kappa+2}<1.
\end{equation}
Evaluating $\eqref{def: input solution set}$ for matrix sets $\bm{\mathcal{A}}$ and $\bm{\mathcal{B}}$ yields
\begin{equation}
\label{def: time-point particular solution}
     \mathcal{P}(t)\subseteq \bigg(\hspace{-2pt} \dsum_{i=0}^{\kappa}\hspace{-2pt}\frac{\bm{\mathcal{A}}^{i}t^{i+1}}{(i+1)!}~ \bm{\mathcal{B}}\hspace{1pt}
     \textproc{fresh}(\mathcal{U})\bigg) 
     \oplus \big ( t \, \bm{\mathcal{E}}_{t} \, \textproc{zonotope}(\bm{\mathcal{B}}\,\mathcal{U})\big),
\end{equation}
where we use exact sum to keep the dependencies between the matrix zonotopes $\bm{\mathcal{A}}$ and $\bm{\mathcal{B}}$. At the same time, we apply the operator $\textproc{fresh}(\mathcal{U})$ to destroy the dependencies between the input sets $\mathcal{U}$, which ensures that we maintain the over-approximate nature of $\eqref{def: input solution set}$. Moreover, based on $\eqref{def: time-point particular solution}$, the particular solution for the initial time interval is given by
\begin{equation}
\label{def: time-continuous particular solution}
     \mathcal{P}(\tau_0)\subseteq \bigg( \hspace{-2pt} \dsum_{i=0}^{\kappa}\hspace{-2pt}\frac{\bm{\mathcal{A}}^{i}\bm{\mathcal{T}}^{i+1}}{(i+1)!}~ \bm{\mathcal{B}}\hspace{1pt} 
     \textproc{fresh}(\mathcal{U})\bigg)
     \oplus \big (\Delta t \, \bm{\mathcal{E}}_{\Delta t} \, \textproc{zonotope}(\bm{\mathcal{B}}\,\mathcal{U})\big),
\end{equation}
where the matrix zonotope $\bm{\mathcal{T}}$ from $\eqref{def: time matrix zonotope}$ equivalently represents the time interval $\tau_0 = [0, \Delta t].$ Given the particular solution for the first time interval, the particular solutions for the remaining time intervals can according to \cite[Prop.~3.2]{Althoff2010a} be computed by the following propagation scheme:
\begin{equation}
\label{def: particular solution recurrence relation}
\mathcal{P}_{A, B}(\tau_k) = e^{A\Delta t}\mathcal{P}_{A, B}(\tau_{k-1}) \oplus \mathcal{P}_{A, B}(\Delta t),
\end{equation}
which considers numerical matrices $A$ and $B$.
Evaluating $\eqref{def: particular solution recurrence relation}$ for matrix sets $\bm{\mathcal{A}}$ and $\bm{\mathcal{B}}$ yields
\begin{equation*}
\begin{split}
 \mathcal{P}(\tau_k) \subseteq~&e^{\bm{\mathcal{A}}\Delta t}~\mathcal{P}(\tau_{k-1})\boxplus \textproc{fresh}\big(\mathcal{P}(\Delta t), \bm{\mathcal{A}}, \bm{\mathcal{B}}\big),
\end{split}
\end{equation*}
where we again use the exact sum to preserve the dependencies between the matrix sets $\bm{\mathcal{A}}$ and $\bm{\mathcal{B}}$, but destroy dependencies between the particular solutions at different points in time using $\textproc{fresh}$ to realize the Minkowski sum in $\eqref{def: particular solution recurrence relation}$.

\vspace{\secminus}
\subsection{Reachability Algorithm}
Combining the homogeneous solution from Sec. \ref{section:homogeneous solution} and the particular solution from Sec. \ref{section:partiuclar solution}, we obtain the reachability algorithm summarized in Alg. \ref{alg:mainReach}:
\begin{algorithm}[h!tb]
    \caption{Reachability algorithm (constant parameters)}
    \label{alg:mainReach}
    
    \begin{FlushLeft}
        \textbf{Require:} Linear parametric system $\eqref{def:probdef}$ with matrix zonotopes $\bm{\mathcal{A}}$ and $\bm{\mathcal{B}}$, initial set $\mathcal{X}_{0}$, input set $\mathcal{U}$, time horizon $t_{end}$, time step size $\Delta t$, Taylor order $\kappa$, desired zonotope order $\rho_d$
    \end{FlushLeft}
    \begin{FlushLeft}
        \textbf{Ensure:} Reachable set $\mathcal{R}([0, t_{end}])$
    \end{FlushLeft}

    \begin{algorithmic}[1]
    \State $\bm{\mathcal{T}}\gets\eqref{def: time matrix zonotope}$
    \vspace{1.2pt}
    \State $\mathcal{H}(\tau_0) \gets e^{\bm{\mathcal{A}}\bm{\mathcal{T}}} \, \mathcal{X}_0$ \label{algo1: initizal homo} \Comment{see Prop.~\ref{prop:raise_to_e}}
    \vspace{1.2pt}
    \State $\mathcal{P}(\tau_0) \gets \eqref{def: time-continuous particular solution},~\mathcal{P}(\Delta t) \gets \eqref{def: time-point particular solution}$
    \label{algo1: initizal input}
    \vspace{1.2pt}
    \State $\mathcal{R}(\tau_0) \gets \mathcal{H}(\tau_0) \boxplus \mathcal{P}(\tau_0)$
    \label{alg1:homo add input}
    \vspace{1.2pt}
    \For{$k = 1, \ldots, (t_{end} / \Delta t)-1$}
        \vspace{1.2pt}
        \State $\mathcal{H}(\tau_k) \gets \textproc{reduce}\big(e^{\bm{\mathcal{A}}\Delta t}\mathcal{H}(\tau_{k-1}),~\rho_d \big)$
        \vspace{1.2pt}
        \State $\mathcal{P}(\Delta t)\gets \textproc{fresh}\big(\mathcal{P}(\Delta t),\bm{\mathcal{A}}, \bm{\mathcal{B}} \big) $
        \vspace{1.2pt}
        \State $\mathcal{P}(\tau_k) \gets \textproc{reduce}\big( e^{\bm{\mathcal{A}}\Delta t}\mathcal{P}(\tau_{k-1})\boxplus \mathcal{P}(\Delta t),~\rho_d \big)$
        \vspace{1.2pt}
        \State $\mathcal{R}(\tau_k) \gets \mathcal{H}(\tau_k) \boxplus \mathcal{P}(\tau_k)$
        \label{alg1: homo add input in rounds}
        \vspace{1.2pt}
    \EndFor
    \vspace{1.2pt}
    \State $\mathcal{R}([0, t_{end}]) \gets \bigcup_{k=0}^{(t_{end} / \Delta t) - 1}\mathcal{R}(\tau_k)$
    \end{algorithmic}
\end{algorithm}
\begin{figure}[ht]
   \centering
   \setlength{\belowcaptionskip}{-12pt}
   \includegraphics[width=0.5\columnwidth]{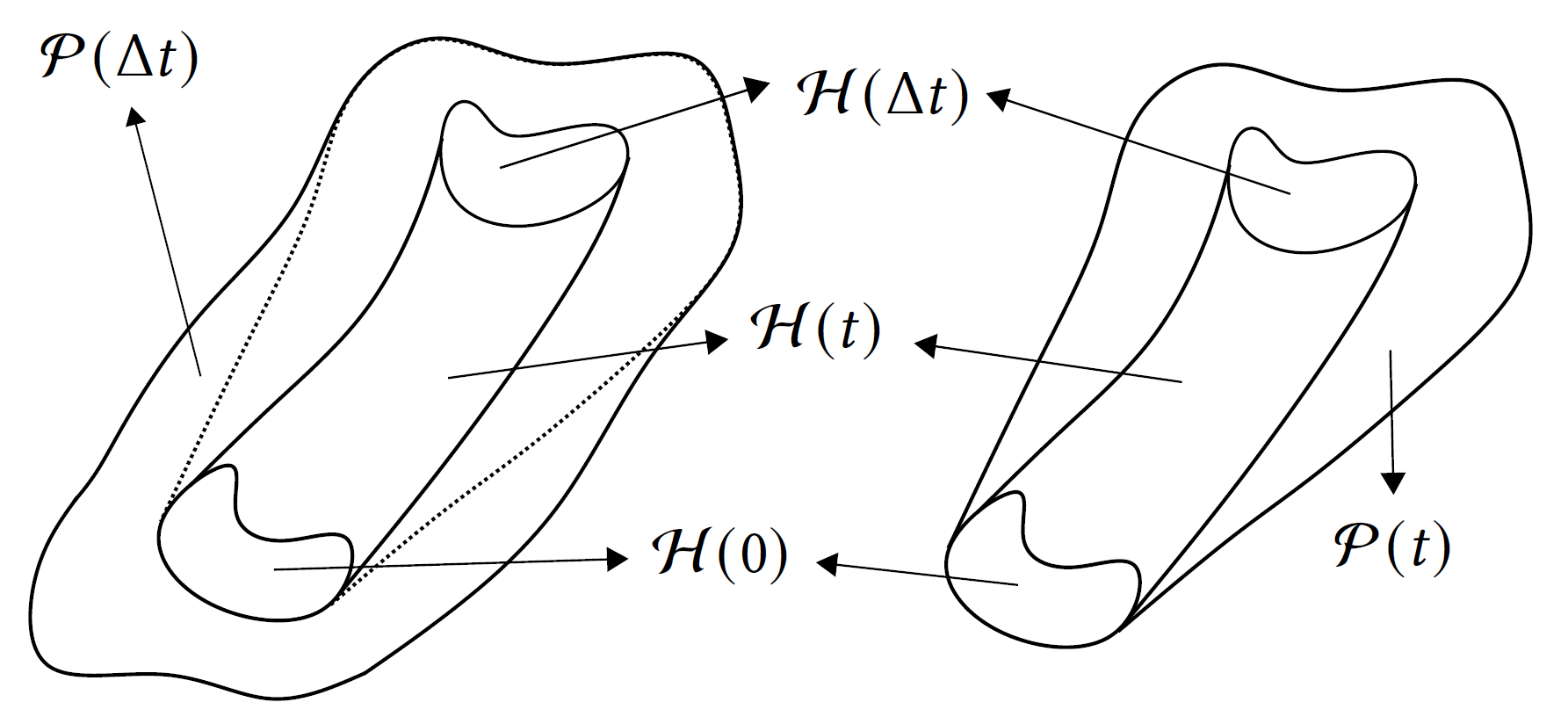}
   \vspace{-8pt}
   \caption{Comparison between combining the homogeneous solution $\mathcal{H}(\tau_0)$ and particular solution $\mathcal{P}(\tau_0)$ via Minkowski sum (left) and via exact sum (right).}
   \label{fig:homo and particualr solution direct addition}
\end{figure}

The algorithm takes advantage of dependency preservation for polynomial zonotopes and is therefore able to compute very tight non-convex enclosures of the reachable set. In Line $\ref{algo1: initizal homo}$ and $\ref{algo1: initizal input}$, the homogeneous solution and particular solution for the first time interval $\tau_0 = [0, \Delta t]$ are calculated. Since both of them preserve dependency on time, we can combine them using the exact sum in Line $\ref{alg1:homo add input}$. As visualized in Fig. $\ref{fig:homo and particualr solution direct addition}$, this results in a much tighter set compared to adding the particular solution $\mathcal{P}(\Delta t)$ at the end of the time interval to every part of the homogeneous solution, which is the technique applied by other reachability algorithms \cite{Althoff2011b,Frehse2011}. Moreover, the property of time preservation not only holds for the initial step, but also for all consecutive steps since forward-propagation via multiplication with $e^{\bm{\mathcal{A}}\Delta t}$ does not destroy dependency on time. Therefore, we can also combine the propagated homogeneous and particular solution for all remaining time intervals with the exact sum in Line $\ref{alg1: homo add input in rounds}$. Moreover, since dependencies on time are preserved, we can extract the reachable set $\mathcal{R}(t)$ for a specific time point $t \in \tau_k$ from the reachable set $\mathcal{R}(\tau_k)$ of the corresponding time interval $\tau_k$ using the \textproc{eval} operation: 
\begin{equation*}
    \forall t \in \tau_k: \quad \mathcal{R}(t) = \textproc{eval}\big ( \mathcal{R}(\tau_k), id_t, \alpha_t \big), \quad \alpha_t = 2 \, \frac{t - k\,\Delta t}{\Delta t} - 1,
\end{equation*}
where $id_t$ is the unique id of the dependent factor that represents time and $\alpha_t \in [-1,1]$ maps the time $t$ to the corresponding factor. 

Finally, to reduce the over-approximation introduced by the remainder $\bm{\mathcal{E}}$ in Prop.~\ref{prop:raise_to_e} one can evaluate additional Taylor terms with interval arithmetic \cite{Jaulin2006} before truncating the Taylor series. Moreover, one can apply restructuring as described in \cite[Sec.~III.A]{Kochdumper2019a} to move large independent generators to the dependent generators, which improves accuracy.

\vspace{\secminus}
\section{Applications and Reachability Evaluation}
\label{sec:Application}
There exist several extensions and applications for the reachability algorithm from the previous section, which we present next. In particular, we first show how to extend the algorithm to linear systems with time-varying parameters (Sec. $\ref{sec:time varying parameter}$) and linear time-varying systems (Sec. $\ref{sec:time varying system}$). Next, we demonstrate the advantages our approach has for linear time-invariant systems without uncertain parameters (Sec. $\ref{sec:LTI single step}$), before we explain how it can be applied to hybrid systems (Sec. $\ref{sec:hybrid system}$) and nonlinear systems (Sec. $\ref{sec:nonlinear system}$).

\vspace{\secminus}
\subsection{Time-Varying Parameters}
\label{sec:time varying parameter}
While our basic reachability algorithm presented in Sec. $\ref{sec:ReachabilityAnalysis}$ considers systems where the parameters are constant over time, it can be easily extended to also handle time-varying parameters. 
For this, we simply have to replace the state transition matrix $e^{\bm{\mathcal{A}} t}$ for the case with constant parameters by an enclosure of state transition matrix $\bm{\mathcal{M}}(t)$ for the case with time varying parameters. According to \cite[Thm.~1]{Althoff2011a}, such an enclosure can be computed as:
\begin{equation}
\begin{split}
    &\bm{\mathcal{M}}(t) \subseteq \bigoplus_{i = 0}^{\kappa}\overline{\bm{\mathcal{M}}}_i(t)~\oplus~\bm{\mathcal{E}}_t, \quad \overline{\bm{\mathcal{M}}}_i(t) = \frac{t^i}{i!}\textproc{conv}\big(\textproc{fresh}^i(\bm{\mathcal{A}})\big),
\end{split}
\label{def:transition matrix}
\end{equation}
where $\kappa \geq 1$ is the user-defined Taylor order, the interval matrix $\bm{\mathcal{E}}_t$ is defined as in \eqref{def: remainder at time t}, and operation $\textproc{conv}(\bm{\mathcal{A}})$ represents the convex hull of a matrix set $\bm{\mathcal{A}}$ which is defined as
\begin{equation}
    \label{def:convex hull}
    \textproc{conv}(\bm{\mathcal{A}}) = \big\{\lambda \, A_1 + (1 - \lambda) \, A_2~\big|~ A_1, A_2 \in \bm{\mathcal{A}}, ~ \lambda \in [0, 1] \big\}.
\end{equation}
Moreover, the notation 
\begin{equation*}
\begin{split}
    \textproc{fresh}^i(\bm{\mathcal{A}}) = \underbrace{\textproc{fresh}(\bm{\mathcal{A}}) ~ \ldots ~ \textproc{fresh}(\bm{\mathcal{A}})}_{i}
\end{split}
\end{equation*}
emphasizes that each matrix zonotope is equipped with different unique identifiers before the multiplication\footnote{Note that the notation $\textproc{fresh}^i(\bm{\mathcal{A}})$ is different from $(\textproc{fresh}(\bm{\mathcal{A}}))^i$, where the latter one represents $\bm{\mathcal{B}}^i$ with $\bm{\mathcal{B}} = \textproc{fresh}(\bm{\mathcal{A}})$. \label{footnote2}}. 
For our reachability algorithm, we require the following proposition:
\begin{proposition}
\label{prop: convex hull computation}
(Multiplication with Convex Hull) Given a matrix zonotope $\bm{\mathcal{A}} \subset \mathbb{R}^{n \times n}$ and a polynomial zonotope $\mathcal{PZ}\subset \mathbb{R}^{n}$, the multiplication between ${\normalfont \textproc{conv}}(\bm{\mathcal{A}})$ and $\mathcal{PZ}$ can be computed as
\begin{equation*}
\begin{split}
    {\normalfont \textproc{conv}}\big(\bm{\mathcal{A}} \big) \, \mathcal{PZ} =
    \bm{\mathcal{A}} ~ \bm{\mathcal{L}}_1 \, \mathcal{PZ}~\boxplus ~{\normalfont \textproc{fresh}}(\bm{\mathcal{A}}) ~ \bm{\mathcal{L}}_2 ~ \mathcal{PZ}
\end{split}    
\end{equation*}
with
\begin{equation}
    \label{def: lambda matrix zonotope}
    \bm{\mathcal{L}}_1 =  \big\langle 0.5 \, I_n, 0.5 \, I_n, {\normalfont\textproc{uniqueID}}(1)\big \rangle_{MZ},~~\bm{\mathcal{L}}_2 = I_n - \bm{\mathcal{L}}_1
\end{equation}
which can be evaluated using Prop. ~\ref{prop:mtimes}.
\end{proposition}
\begin{proof}
According to the definition of the convex hull of a matrix set in \eqref{def:convex hull} we have
\begin{equation*}
    \begin{split}
        \textproc{conv}(\bm{\mathcal{A}}) & = \big\{\lambda \, A_1 + (1 - \lambda) \, A_2~\big|~ A_1, A_2 \in \bm{\mathcal{A}}, ~ \lambda \in [0, 1] \big\} \\
        & = \bm{\mathcal{A}} ~ \bm{\mathcal{L}}_1~\boxplus ~\textproc{fresh}(\bm{\mathcal{A}}) \, \big(I_n - \bm{\mathcal{L}}_1 \big),
    \end{split}
\end{equation*}
where the matrix zonotope $\bm{\mathcal{L}}_1$ defined as in \eqref{def: lambda matrix zonotope} represents the domain $\lambda \in [0,1]$. While we need to apply the \textproc{fresh} operation to $\bm{\mathcal{A}}$ to ensure that we realize two independent matrices $A_1,A_2 \in \bm{\mathcal{A}}$ as demanded by the definition in \eqref{def:convex hull}, we use the exact sum to keep the dependencies between $\bm{\mathcal{L}}_1$ and $\bm{\mathcal{L}}_2$.
\end{proof}

Combining \eqref{def:transition matrix}, Prop. $\ref{prop: convex hull computation}$, and the concept of preserving dependency on time used in Sec. $\ref{sec:ReachabilityAnalysis}$, the homogeneous solution for the first time interval can be calculated as follows:
\begin{align}
     \mathcal{H}(\tau_0) &= \bm{\mathcal{M}}(\tau_0) \, \mathcal{X}_0\\ &\subseteq \big(\bm{\mathcal{E}}_{\Delta t}\textproc{zonotope}(\mathcal{X}_0) \big) ~ \oplus  \bigg( \dsum_{i = 0}^{\kappa}  \frac{1}{i!}\textproc{fresh}^i(\bm{\mathcal{A}}) \, \bm{\mathcal{L}}_1 \, \bm{\mathcal{T}}^i ~ \mathcal{X}_0 \boxplus \frac{1}{i!}\textproc{fresh}^i(\bm{\mathcal{A}}) \, \bm{\mathcal{L}}_2 \, \bm{\mathcal{T}}^i ~ \mathcal{X}_0 \bigg),
\label{comp:time interval transition formula}
\end{align}
where $\bm{\mathcal{L}}_1$ and $\bm{\mathcal{L}}_2$ are given by $\eqref{def: lambda matrix zonotope}$, $\bm{\mathcal{T}}$ is given by $\eqref{def: time matrix zonotope}$, and $\bm{\mathcal{E}}_{\Delta t}$ is given by \eqref{def: remainder at time t}.
Based on \cite[Thm. 2]{Althoff2011a}, we can similarly enclose the particular solution at time $t$ by:
\begin{equation}
\begin{split}
\label{def:time varying input}
& \mathcal{P}(t) \subseteq \bigg( \dsum_{i=0}^{\kappa} \frac{t^{i+1}}{(i+1)!}\textproc{fresh}^i(\bm{\mathcal{A}}) \, \bm{\mathcal{L}}_1 \, \textproc{fresh}(\bm{\mathcal{B}} \, \mathcal{U}) ~ \boxplus \\
&  \frac{t^{i+1}}{(i+1)!} \textproc{fresh}^i(\bm{\mathcal{A}}) \, \bm{\mathcal{L}}_2 \, \textproc{fresh}(\bm{\mathcal{B}} \, \mathcal{U})\bigg) \hspace{-1pt} \oplus \hspace{-1pt} \bigg ( \frac{t}{\kappa + 2} \,\bm{\mathcal{E}}_{t} \, \big \langle \widetilde{u},[~] \big\rangle_Z \bigg ),
\end{split}
\end{equation}
where $\widetilde{u} = [\widetilde{u}_1 ~ \ldots ~ \widetilde{u}_n]^T$ with $\widetilde{u}_j = \max\{|u_{(j)}|~ | ~u\in ~\bm{\mathcal{B}} \,\mathcal{U}\}$, $j = 1,\dots,n$. Moreover, based on \eqref{def:time varying input} the particular solution for the first time interval is then given as
\begin{equation}
\begin{split}
\label{def:time varying input interval}
\mathcal{P}(\tau_0) \subseteq & ~\big( \Delta t /(\kappa + 2) \,\bm{\mathcal{E}}_{\Delta t} \, \langle \widetilde{u},[~] \rangle_Z \big) ~ \oplus \\ 
& ~\bigg( \dsum_{i=0}^{\kappa} \frac{1}{(i+1)!}\textproc{fresh}^i(\bm{\mathcal{A}}) \, \bm{\mathcal{T}}^{i+1} \, \bm{\mathcal{L}}_1 \, \textproc{fresh}(\bm{\mathcal{B}} \, \mathcal{U}) \\
& \hspace{20pt} \boxplus\frac{1}{(i+1)!} \textproc{fresh}^i(\bm{\mathcal{A}}) \, \bm{\mathcal{T}}^{i+1} \, \bm{\mathcal{L}}_2 \, \textproc{fresh}(\bm{\mathcal{B}} \,\mathcal{U})\bigg).
\end{split}
\end{equation}
Finally, a tight enclosure of the reachable set can be computed using Alg.~\ref{alg:tvReach}, which is a modification of \cite[Alg. 1]{Althoff2011a}.
Note that since we have to destroy dependencies on $\bm{\mathcal{A}}$ and $\bm{\mathcal{B}}$ to ensure correctness, we can use a simpler propagation scheme compared to Alg.~\ref{alg:mainReach}.

\begin{algorithm}[h!tb]
    \caption{Reachability algorithm (time-varying parameters)}
    \label{alg:tvReach}
    \begin{FlushLeft}
        \textbf{Require:} Linear system $\eqref{def:probdef}$ with time-varying parameters represented by the matrix zonotopes $\bm{\mathcal{A}}$ and $\bm{\mathcal{B}}$, initial set $\mathcal{X}_0$, input set $\mathcal{U}$, time horizon $t_{end}$, time step size $\Delta t$, Taylor order $\kappa$, desired zonotope order $\rho_d$
    \end{FlushLeft}
    \begin{FlushLeft}
        \textbf{Ensure:} Reachable set $\mathcal{R}([0, t_{end}])$
    \end{FlushLeft}

    \begin{algorithmic}[1]
        \State $\mathcal{H}(\tau_0) \gets \bm{\mathcal{M}}(\tau_0)~\mathcal{X}_0$ \Comment{see \eqref{comp:time interval transition formula}}
        \vspace{1.2pt}
        \State $\mathcal{P}(\tau_0) \gets \eqref{def:time varying input interval}$, $~\mathcal{P}(\Delta t) \gets \eqref{def:time varying input}$
        \label{alg2:prepare inputs}
        \vspace{1.2pt}
        \State $\mathcal{R}(\tau_0) \gets \mathcal{H}(\tau_0) \boxplus \mathcal{P}(\tau_0)$
        \For{$k = 1, \ldots, (t_{end} / \Delta t)-1$}
        \vspace{1.2pt}
        \State $\mathcal{R}(\tau_k) \gets \textproc{reduce}\big(\bm{\mathcal{M}}(\Delta t) \, \mathcal{R}(\tau_{k-1}) \oplus \mathcal{P}(\Delta t), \rho_d \big)$
        \label{alg2:propogation}
        \vspace{1.2pt}
        \EndFor
        \vspace{1.2pt}
        \State $\mathcal{R}([0, t_{end}]) \gets \bigcup_{k=0}^{(t_{end} / \Delta t) - 1}\mathcal{R}(\tau_k)$
    \end{algorithmic}
\end{algorithm}

\vspace{\secminus}
\subsection{Linear Time-Varying Systems}
\label{sec:time varying system}

Based on the reachability algorithm for linear systems with time-varying parameters presented in the previous subsection, we can also compute tight enclosures of reachable sets for time-varying linear systems. In particular, given a time-varying linear system
\begin{equation*}
    \dot x(t) = A(t) \, x(t) + B(t) \, u(t) 
\end{equation*}
with $A:~\mathbb{R}_{\geq 0} \to \mathbb{R}^{n \times n}$ and $B:~\mathbb{R}_{\geq 0} \to \mathbb{R}^{n \times m}$, we can for each reachability time step $\tau_k$ abstract it by the linear parametric system
\begin{equation*}
    \forall t \in \tau_k: ~~ \dot x(t) \in \bm{\mathcal{A}} \, x(t) + \bm{\mathcal{B}} \, u(t),
\end{equation*}
where 
\begin{equation}
    \bm{\mathcal{A}} = \big\{ A(t) ~\big|~ t \in \tau_k \big\}, ~~ \bm{\mathcal{B}} = \big\{ B(t) ~\big|~ t \in \tau_k \big\}.
    \label{eq:MatTimeVarying}
\end{equation}
To tightly enclose the matrix sets $\bm{\mathcal{A}}$ and $\bm{\mathcal{B}}$ in \eqref{eq:MatTimeVarying} by matrix zonotopes, range bounding using affine arithmetic \cite{deFigueiredo2004} can be applied. Finally, we can simply use Alg.~\ref{alg:tvReach} to compute the reachable set, where we update the matrices $\bm{\mathcal{A}}$ and $\bm{\mathcal{B}}$ according to \eqref{eq:MatTimeVarying}
in each time step. 
We need to use the algorithm for time-varying parameters instead of Alg.~\ref{alg:mainReach} since $A(t)$ and $B(t)$ change over time.
\begin{figure}[ht]
   \centering
   \setlength{\belowcaptionskip}{-10pt}
   \includegraphics[width=0.7\columnwidth]{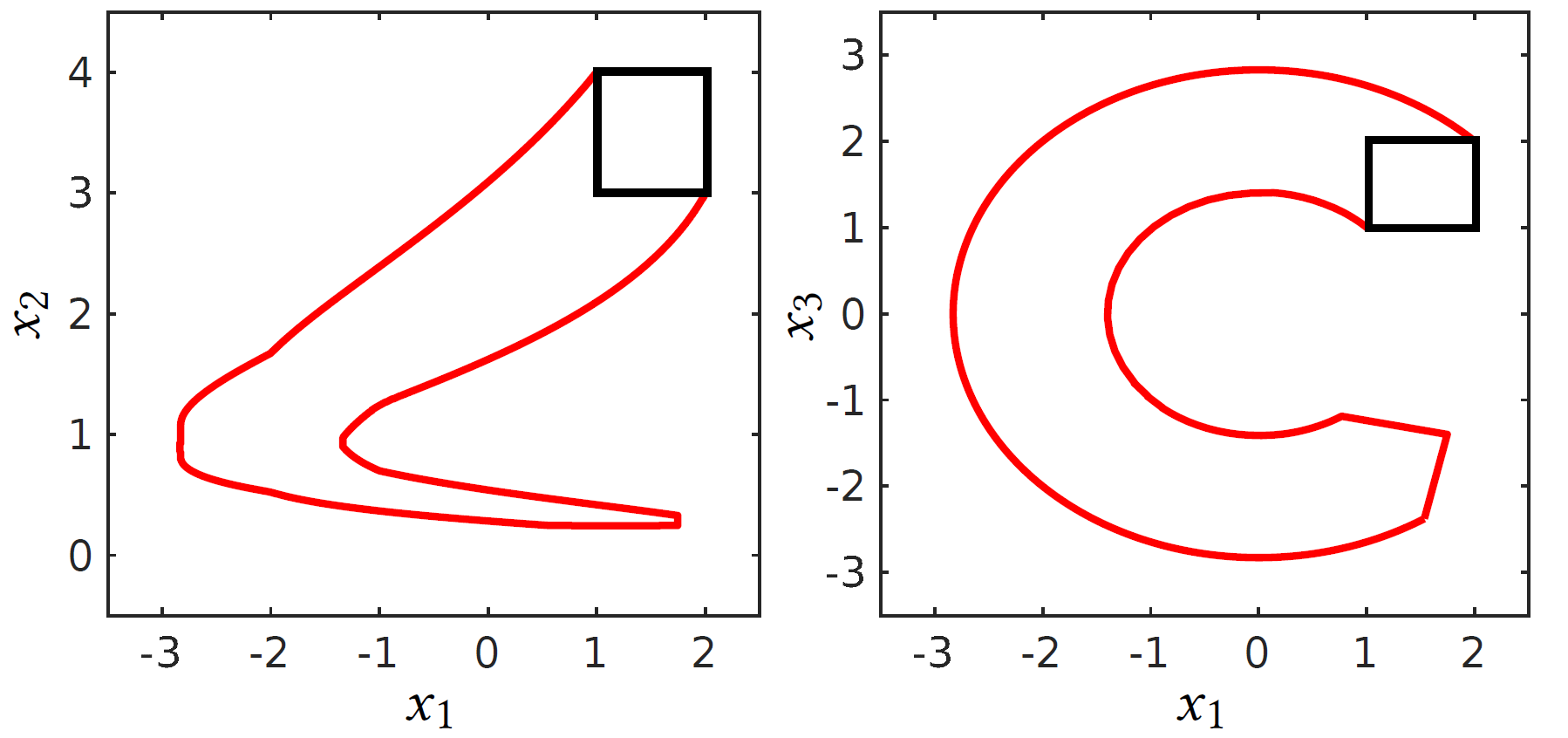}
   \vspace{-8pt}
   \caption{Reachable set for the linear system from Example~\ref{example:linear} (red), where the initial set is depicted in white with a black border.}
    \label{fig:linear}
\end{figure}
\vspace{\secminus}
\subsection{Linear Time-Invariant Systems}
\label{sec:LTI single step}
Even though reachability analysis for linear time-invariant systems $\dot x(t) = A\, x(t) + B \, u(t)$ is a well studied topic, our approach provides unique advantages for this problem. As summarized in a recent survey \cite{Forets2022}, existing methods compute the homogeneous time interval solution either via optimization \cite{Frehse2011} or by calculating the convex hull of the time-point solutions \cite{Girard2005,LeGuernic2010}, where an additional correction term is used to account for the trajectory curvature. All of these methods are only accurate if the time step size $\Delta t$ is small. With our approach, on the other hand, we can compute the homogeneous time interval solution very elegantly by simply evaluating $\mathcal{H}(\tau_0) = e^{A\bm{\mathcal{T}}} \, \mathcal{X}_0$ with $\bm{\mathcal{T}}$ from \eqref{def: time matrix zonotope} using Prop. $\ref{prop:raise_to_e}$. Moreover, the representation size required to compute the set with an arbitrary user-defined precision remains moderate even for extremely large time intervals, as we  demonstrate by the following lemma:
\begin{lemma}
    Fix an arbitrary error bound \textcolor{red}{$\psi > 0$}.
    Then the number of Taylor terms $\kappa$ required to ensure that the over-approximation error for the enclosure of the set $\mathcal{H}([0,\Delta t]) = e^{A\bm{\mathcal{T}}} \, \mathcal{X}_0$ (calculated using Prop.~\ref{prop:raise_to_e}) stays below $\psi$ only grows linearly with respect to the time step size $\Delta t$. 
    \label{lemma:time}
\end{lemma}
\begin{proof}
According to Prop.~\ref{prop:raise_to_e}, the enclosure is calculated as follows:
\begin{equation}
        e^{A \bm{\mathcal{T}}}\, \mathcal{X}_0 \subseteq \underbrace{\bigg(\dsum_{i=0}^{\kappa}\frac{A^i\bm{\mathcal{T}}^i}{i!}~\mathcal{X}_0\bigg) }_{\text{exact}} \oplus \underbrace{\Big(\bm{\mathcal{E}}~ \textproc{zonotope}(\mathcal{X}_0)\Big)}_{\substack{\text{over-approximation} \\ \vspace{-4pt} \\ \text{error}}}.
\label{eq:proof1}
\end{equation}
Since the first term is evaluated exactly using the exact sum, the over-approximation error is given by the second term. 
Note that $||\bm{\mathcal{T}}||_{\infty} = \Delta t$ holds since $\bm{\mathcal{T}} = [0, \Delta t] \cdot I_n$ according to \eqref{def: time matrix zonotope}. Therefore,
 the interval matrix $\bm{\mathcal{E}}$ is according to Prop.~\ref{prop:raise_to_e} defined as
\begin{equation}
        \bm{\mathcal{E}} = \langle-\boldsymbol{1}_{n \times n}, \boldsymbol{1}_{n \times n}\rangle_{IM} \frac{|| A ||_{\infty}^{\kappa+1} \Delta t ^{\kappa+1}}{(\kappa+1)!}  \frac{1}{1-\epsilon}  
\label{eq:proof2}
\end{equation}
with the convergence requirement for the remainder term, such that
\begin{equation*}
    \epsilon = \frac{\| A \|_{\infty} \Delta t}{\kappa + 2} < 1. \label{eq:proof3}
\end{equation*}
For any $\psi > 0 $, define $\zeta = \min(\frac{\psi}{2},0.5)$. By selecting $\kappa$ to satisfy
\begin{equation}
    \kappa > \frac{\| A \|_{\infty} \Delta t}{\zeta \frac{1}{e}} - 1, \label{eq:condition}
\end{equation}
we obtain
\begin{equation*}
    (\kappa + 1)! \overset{(a)}{>} \left(\frac{1}{e}\right)^{\kappa + 1}(\kappa + 1)^{\kappa + 1} \overset{\eqref{eq:condition}}{>}  \frac{\| A \|_{\infty}^{\kappa + 1} \Delta t^{\kappa + 1}}{\zeta^{\kappa + 1}},
\end{equation*}
where (a) follows from Stirling's approximation,
\begin{equation*}
    \frac{1}{e} \cdot (\kappa + 1) < \frac{1}{e}(2\pi (\kappa + 1))^{\frac{1}{2(\kappa + 1)}} \cdot e^{\frac{1}{12(\kappa + 1)^2 + (\kappa + 1)}} \cdot (\kappa + 1) < \sqrt[\kappa + 1]{(\kappa + 1)!}.
\end{equation*}
This leads to that 
\begin{equation*}
    \frac{\| A \|_{\infty}^{\kappa + 1} \Delta t^{\kappa + 1}}{(\kappa + 1)!} < \zeta^{\kappa + 1} < \zeta.
\end{equation*}
Additionally, considering the definition of $\epsilon$, we find
\begin{equation*}
    \epsilon = \frac{\| A \|_{\infty} \Delta t}{\kappa + 2} < \frac{\| A \|_{\infty} \Delta t}{\kappa + 1} \overset{\eqref{eq:condition}}{<}   \frac{\zeta}{e} < \frac{1}{2} < 1,
\end{equation*}
which implies that $\frac{1}{1 - \epsilon} < 2$. As a result, it leads to the desired result
\[
\frac{|| A ||_{\infty}^{\kappa+1} \Delta t ^{\kappa+1}}{(\kappa+1)!}  \frac{1}{1-\epsilon}  < \psi. 
\]

\end{proof}

Since the matrix zonotope $\bm{\mathcal{T}}$ has only a single generator according to \eqref{def: time matrix zonotope}, it holds that the enclosure for $\mathcal{H}(\tau_0) = e^{A\bm{\mathcal{T}}} \, \mathcal{X}_0$ calculated using Prop.~\ref{prop:raise_to_e} has $(\kappa + 1) h$ dependent generators if $\mathcal{X}_0$ has $h$ dependent generators. According to Lemma~\ref{lemma:time}, the number of generators required to enclose the homogeneous solution $\mathcal{H}([0,\Delta t])$ with a certain precision therefore only grows linearly with respect to the time step. 
Consequently, we can make the time step size $\Delta t$ very large (for example, the entire time bound) since the representation size grows slowly.

\begin{example}
    We consider the linear system 
    \begin{equation*}
        \begin{bmatrix} \dot x_1 \\ \dot x_2 \\ \dot x_3 \end{bmatrix} = \begin{bmatrix} 0 & 0 & -0.9 \\  0 & -0.5 & 0 \\ 0.9 & 0 & 0\end{bmatrix} \begin{bmatrix} x_1 \\ x_2 \\ x_3 \end{bmatrix}   
    \end{equation*}
    together with the initial set $\mathcal{X}_0 = [1,2] \times [3,4] \times [1,2]$. As shown in Fig.~\ref{fig:linear}, with our approach we can represent the reachable set for the entire time horizon of $t_{end} = 5$s by a single polynomial zonotope. 
    \label{example:linear}
\end{example}

\vspace{\secminus}
\subsection{Hybrid Systems}
\label{sec:hybrid system}

Hybrid systems divide the state space into several regions with different dynamic equations, where the regions are separated by so-called guard sets. 
\begin{figure}[ht]
   \centering
   \setlength{\belowcaptionskip}{-15pt}
   \includegraphics[width=0.7\columnwidth]{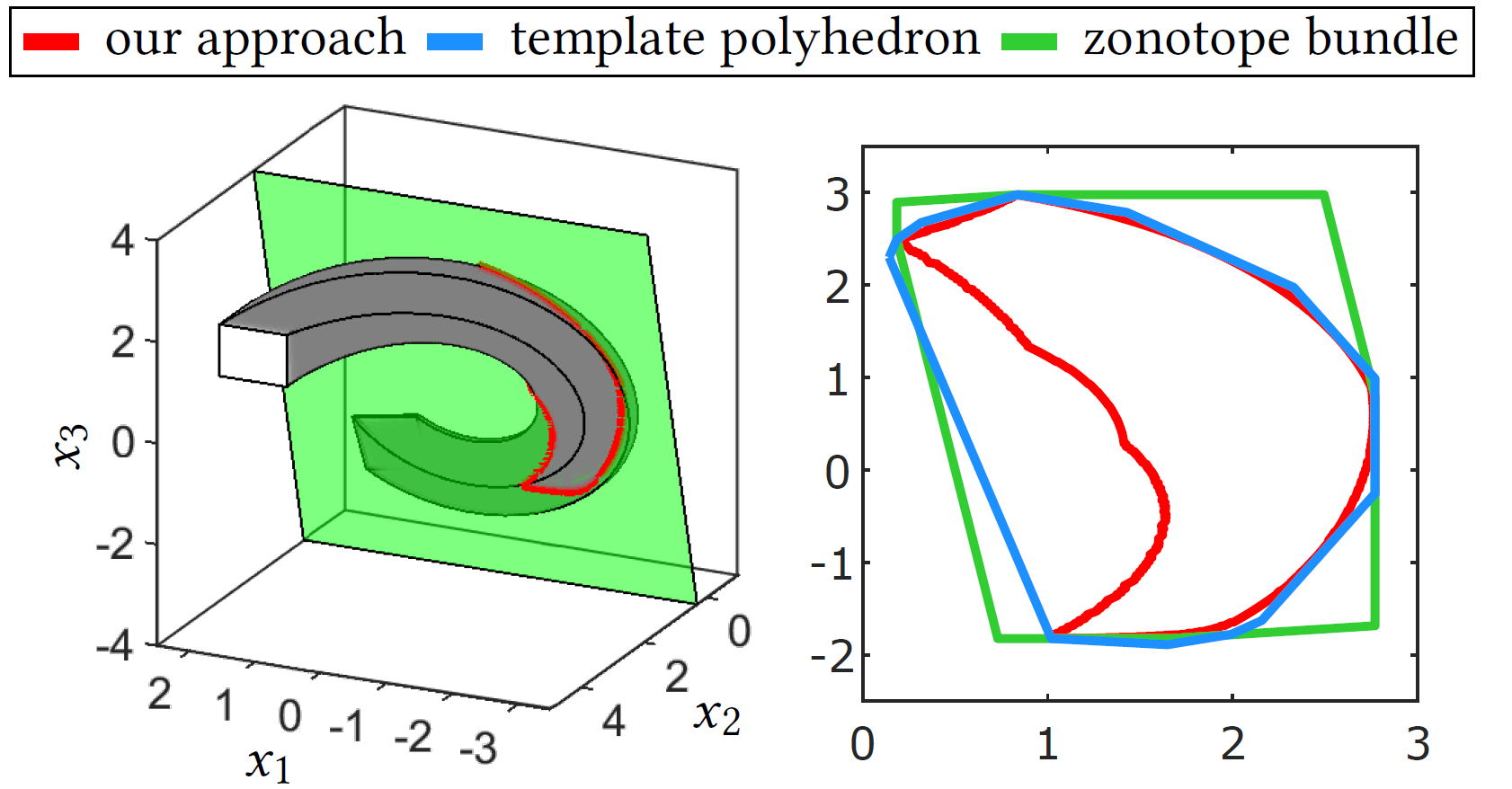}
   \vspace{-8pt}
   \caption{The left plot visualizes the intersection (red) between the reachable set of the linear system from Example~\ref{example:linear} (gray) and the guard set in (\ref{eq:guardSet}) (green), where the initial set is depicted in white with a black border. The resulting guard intersection enclosure computed with different set representations visualized on the plane defined by the guard set is shown on the right.}
    \label{fig:hybrid}
\end{figure}As a simple example for a hybrid system we consider the linear system from Example~\ref{example:linear} with the guard set 
\begin{equation}
    \mathcal{G} = \big\{ x \in \mathbb{R}^3~\big|~ x_{(2)} - 0.2 \, x_{(3)} = 1.1 \big\},
    \label{eq:guardSet}
\end{equation}
which is visualized in Fig.~\ref{fig:hybrid}. The common procedure for reachability analysis of hybrid systems is to use a reachability algorithm for continuous dynamics to compute the reachable set until it has fully crossed the guard set and is therefore located outside of the current region. The initial set for the next region is then determined by calculating the intersection between the reachable set and the guard set, before the whole procedure is repeated until the time horizon is reached. Handling those guard intersections efficiently is one of the main challenges for hybrid system reachability analysis. In particular, most reachability algorithms for continuous systems require a very small time step size $\Delta t$ to provide accurate results. Consequently, often the reachable sets for many consecutive time intervals will intersect the guard set. If the intersection between the guard set and each of these sets would be treated as a separate initial set for the next region, one would have to solve not only one but many reachability problems, which is computationally demanding. Moreover, since usually multiple consecutive guard intersections occur during the overall time horizon, the number of parallel sets would grow exponentially in this case \cite{Duggirala2019}. To avoid these issues, existing approaches for hybrid system reachability unite the sets for different time intervals by enclosing the corresponding guard intersections by convex set representations such as template polyhedra \cite{Girard2008}, zonotopes \cite{Althoff2012a}, or zonotope bundles \cite{Althoff2011f}. However, as shown in Fig.~\ref{fig:hybrid}, even for very simple linear dynamics the intersection between the reachable set and the guard set often yields a highly non-convex set, so that convex enclosures are often very conservative. 

\begin{figure}[ht]
   \centering
   \setlength{\belowcaptionskip}{-13pt}
   \includegraphics[width=0.7\columnwidth]{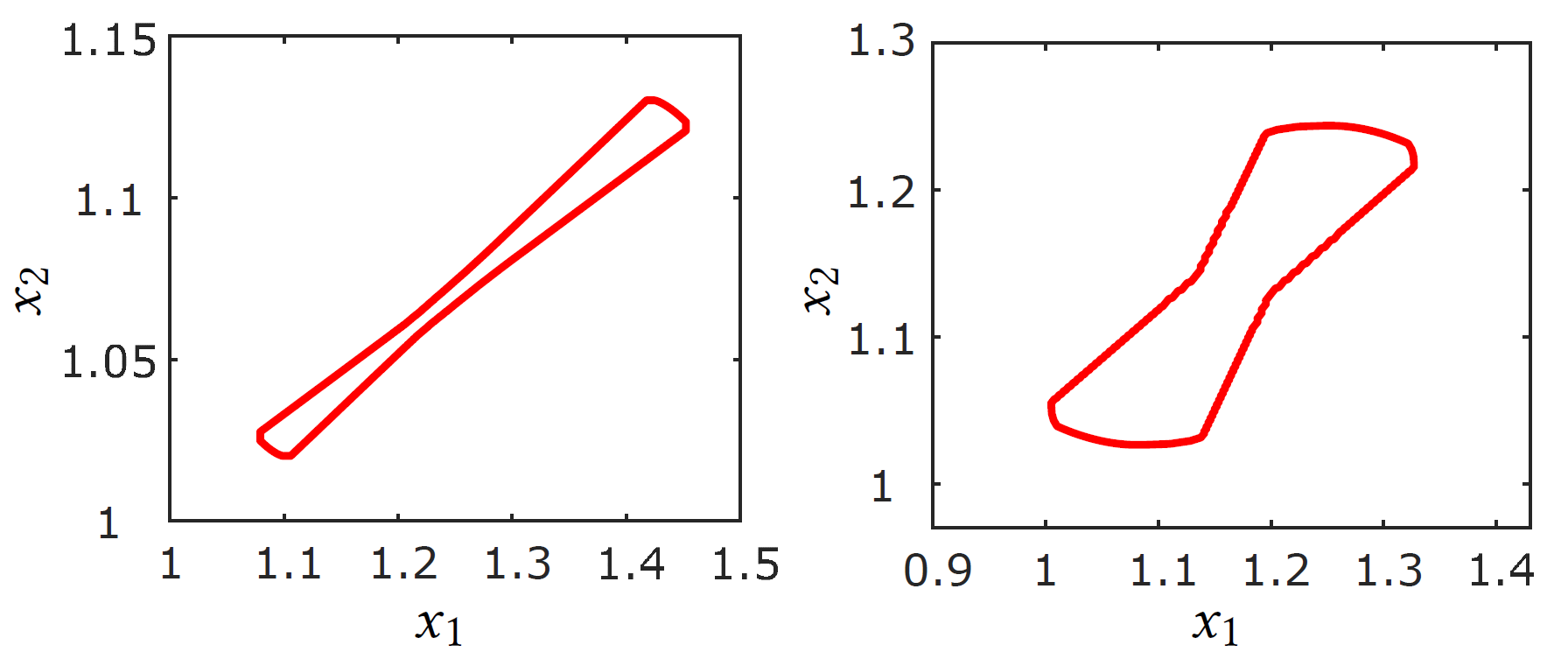}
   \vspace{-8pt}
    \caption{Reachable set enclosure for the Lotka-Volterra system in Example~\ref{ex:nonlinear} at time $t = 0.25$s (left) and time $t = 0.5$s (right).}
    \label{fig:nonlinear}
\end{figure}

As we discussed in Sec.~\ref{sec:LTI single step}, with our reachability approach we can represent the reachable set for the entire time horizon with a single set. Consequently, with our approach one can completely avoid the unification step when computing guard intersections. Moreover, by converting the polynomial zonotope that represents the reachable set to a constrained polynomial zonotope \cite{Kochdumper2020}, the intersection between the reachable set and the guard set can be computed exactly as a single set, even for nonlinear guard set such as polynomial level sets \cite[Prop.~3.2.24]{Kochdumper2022b}. Due to the similarity of polynomial zonotopes and constrained polynomial zonotopes, the reachability algorithm presented in this paper can easily be extended to compute with constrained polynomial zonotopes, so that we can simply use the resulting constrained polynomial zonotope as the initial set for the next region. 
Finally, checking if a polynomial zonotope intersects a guard set can use the computationally efficient polynomial zonotope refinement approach~\cite{Bak2022a}.
%
%

%
\vspace{\secminus}
\subsection{Nonlinear Systems}
\label{sec:nonlinear system}




The approach presented in this paper can also be used to construct a new reachability algorithm for nonlinear systems.
This algorithm linearizes the system in each time step and captures the linearization error with uncertain parameters in order construct a guaranteed enclosure of the reachable set.
In particular, given the nonlinear system
\begin{equation*}
    \dot x(t) = f\big(x(t),u(t)\big)
\end{equation*}
with $f:~\mathbb{R}^n \times \mathbb{R}^m \to \mathbb{R}^n$, we can for each reachability time step $\tau_k$ enclose it by a linear parametric system using a first-order Taylor series expansion of the function $f$:
\begin{equation*}
    \forall t \in \tau_k: ~~ \dot x(t) \in f(x_l,u_l) + \bm{\mathcal{A}} \big( x(t) - x_l \big) + \bm{\mathcal{B}} \big( u(t) - u_l \big), ~ 
\end{equation*}
where
\begin{equation}
\begin{split}
    \bm{\mathcal{A}} = & \bigg \{ \frac{\partial f(x,u)}{\partial x} ~\bigg|~ x \in \mathcal{R}(\tau_k), u \in \mathcal{U} \bigg \}, \\
    \bm{\mathcal{B}} = & \bigg \{ \frac{\partial f(x,u)}{\partial u} ~\bigg|~ x \in \mathcal{R}(\tau_k), u \in \mathcal{U} \bigg \}.
\end{split}
\label{eq:MatNonlinear}
\end{equation}
The expansion points $x_l$ and $u_l$ can be heuristically chosen as $x_l = x_c + 0.5 \, \Delta t f(x_c,u_c)$ and $u_l = u_c$, where $x_c$ is the center of $\mathcal{R}(t_k)$ and $u_c$ is the center of $\mathcal{U}$. To tightly enclose the matrix sets $\bm{\mathcal{A}}$ and $\bm{\mathcal{B}}$ in \eqref{eq:MatNonlinear} by matrix zonotopes, we can apply range bounding using affine arithmetic \cite{deFigueiredo2004}. One problem we are facing here is that we require the time interval reachable set $\mathcal{R}(\tau_k)$ to compute $\bm{\mathcal{A}}$ and $\bm{\mathcal{B}}$, but also require $\bm{\mathcal{A}}$ and $\bm{\mathcal{B}}$ to compute $\mathcal{R}(\tau_k)$. To resolve this mutual dependence, we apply the following strategy \cite{Althoff2013a}: We compute $\bm{\mathcal{A}}$ and $\bm{\mathcal{B}}$ using an estimate $\widehat{\mathcal{R}}(\tau_k)$ for $\mathcal{R}(\tau_k)$, and  then compute $\mathcal{R}(\tau_k)$ using the resulting $\bm{\mathcal{A}}$ and $\bm{\mathcal{B}}$. If the resulting set satisfies $\mathcal{R}(\tau_k) \subseteq \widehat{\mathcal{R}}(\tau_k)$, we have a guaranteed enclosure of the reachable set. If not, we update our estimate by bloating the resulting set $\widehat{\mathcal{R}}(\tau_k) = x_c + \eta \, (\mathcal{R}(\tau_k) - x_c)$, where $\eta > 1$ is a bloating factor and $x_c$ is the center of $\mathcal{R}(\tau_k)$. This process is repeated until $\mathcal{R}(\tau_k) \subseteq \widehat{\mathcal{R}}(\tau_k)$ is satisfied at some point. As an initial estimate $\widehat{\mathcal{R}}(\tau_k)$ we can use the final reachable set from the previous time step $\mathcal{R}(k \Delta t)$. Since containment checks as well as affine arithmetic cannot be directly applied to polynomial zonotopes, we use an interval or oriented hyper-rectangle enclosure of $\mathcal{R}(\tau_k)$ to compute $\bm{\mathcal{A}}$ and $\bm{\mathcal{B}}$. Moreover, we have to use the algorithm for time-varying parameters in Alg.~\ref{alg:tvReach} to compute the reachable set since the abstraction error represented by the matrix sets $\bm{\mathcal{A}}$ and $\bm{\mathcal{B}}$ does not remain constant over time. Let us demonstrate reachability analysis for nonlinear systems by an example:
\begin{figure}
   \centering
   \setlength{\belowcaptionskip}{-15pt}
   
   \begin{subfigure}[b]{0.49\columnwidth}
      \centering
      \includegraphics[width=\linewidth]{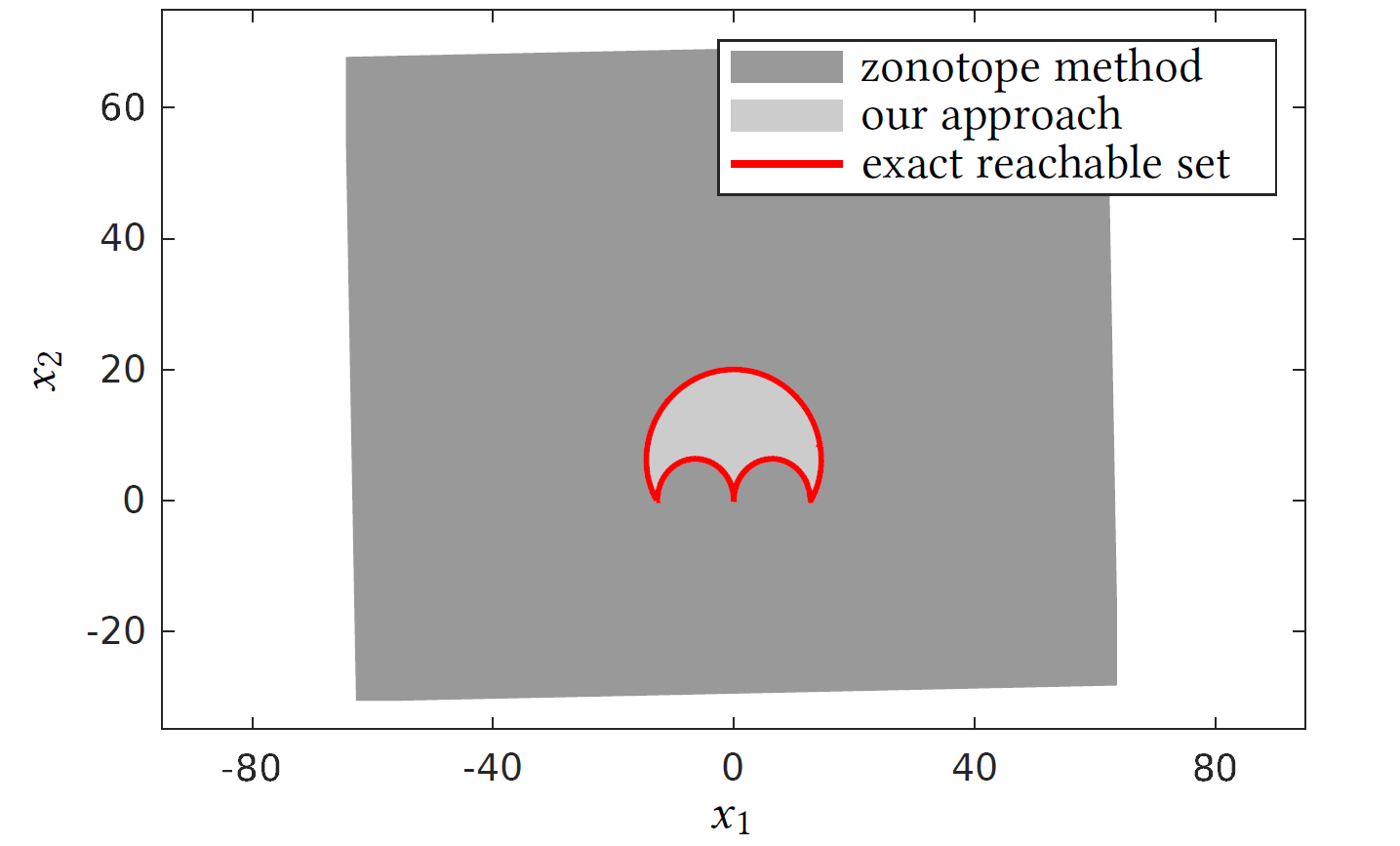}
   \end{subfigure}
   \hfill 
   \begin{subfigure}[b]{0.49\columnwidth}
      \centering
      \includegraphics[width=\linewidth]{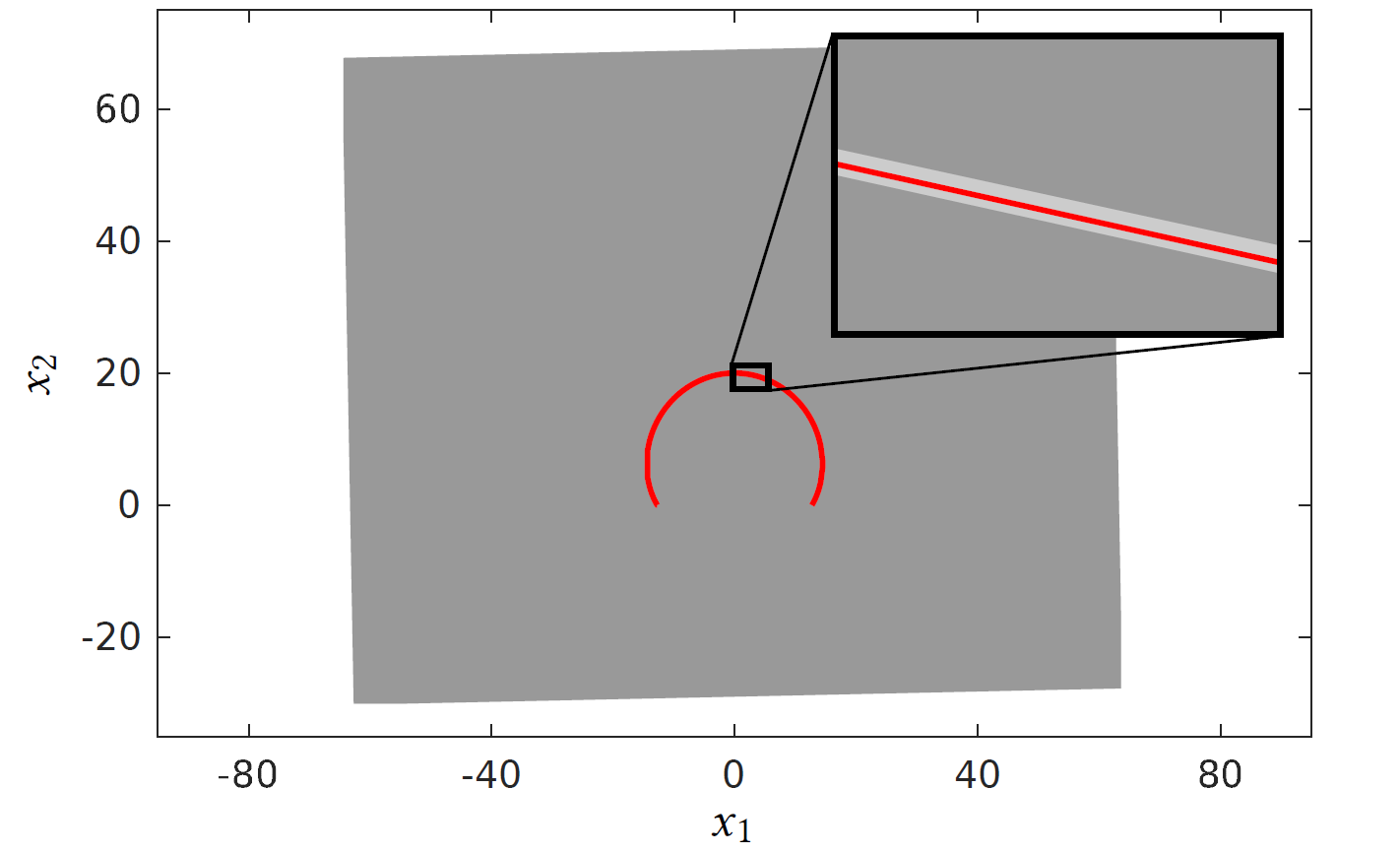}
   \end{subfigure}
   
   \vspace{-10pt}
   \caption{Comparison of reachable set enclosures for the Dubins car computed with our approach and with the zonotope method \cite{Althoff2011b}, where the time interval reachable set is shown on the top and the final reachable set is shown on the bottom.}
   \label{fig:dubin car comparison time interval}
\end{figure}

%
\begin{example}
We consider the Lotka-Volterra system
    \begin{equation*}
    \begin{split}
        \dot x_1 &= 3 \, x_1 - 3 \, x_1 \, x_2 \\
        \dot x_2 &= x_1 \, x_2 - x_2
    \end{split}
    \end{equation*}
    together with  the initial set $\mathcal{X}_0 = [1.1,1.5] \times 1$ and the time horizon $t_{end} = 0.5\si{\second}$. The results visualized in Fig.~\ref{fig:nonlinear} show that our approach is able to compute a non-convex enclosure of the reachable set.
    \label{ex:nonlinear}
\end{example}
\vspace{\secminus}
\label{sec:NumericalEvaluation}

In this section we compare our reachability algorithm to other state of the art methods using two benchmark systems. We implemented our approach in MATLAB, and all computations are carried out on a Apple M2 Pro processor with 16GB memory. 


\begin{figure}
   \centering
   \setlength{\belowcaptionskip}{-15pt}
   \hspace{-8pt}
   \includegraphics[width=0.5\columnwidth]{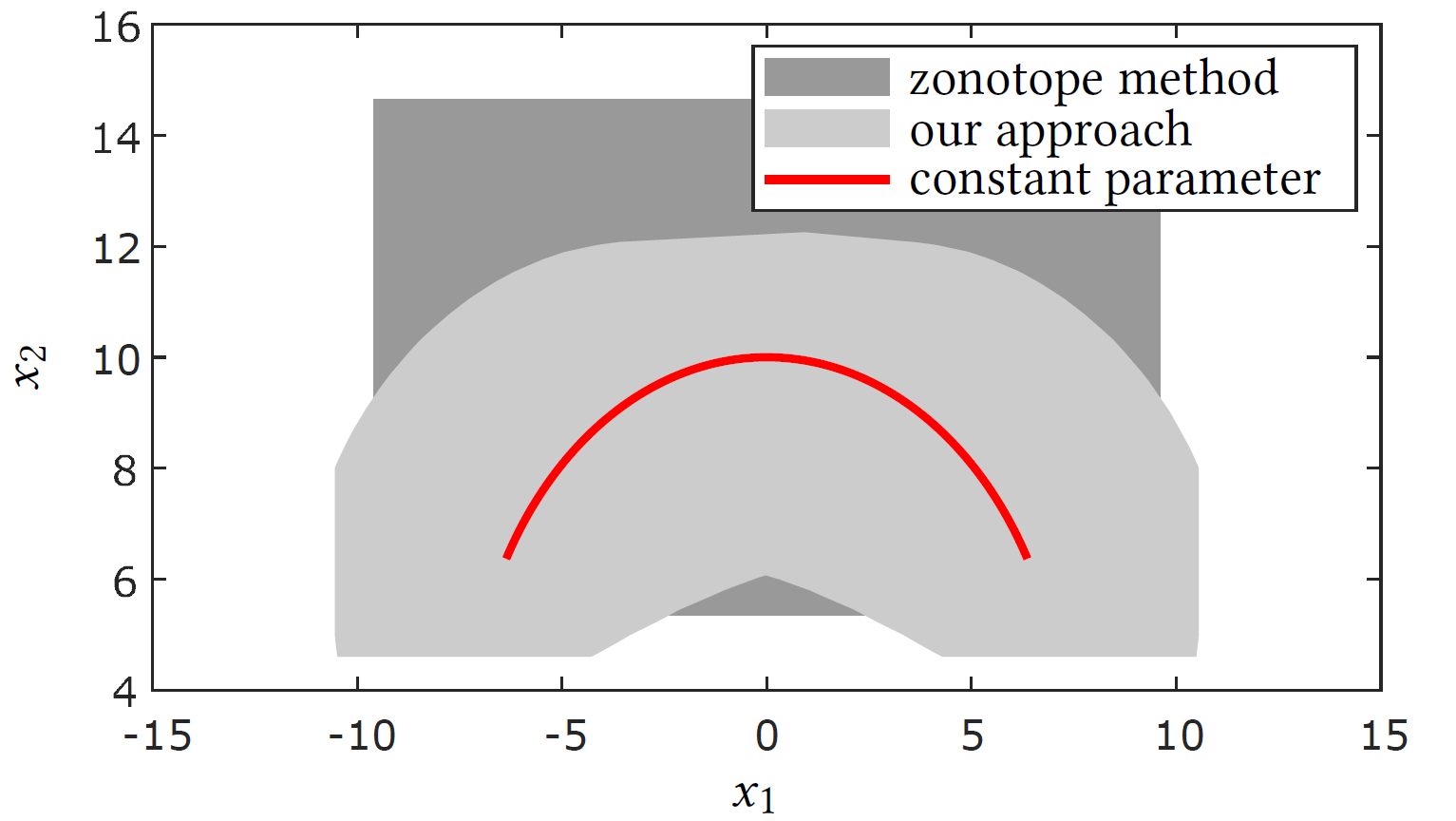}
   \vspace{-10pt}
    \caption{Comparison of enclosures for the final reachable set for the Dubins car with time-varying parameters computed with our approach and with the zonotope method \cite{Althoff2011a}, where the exact reachable set for the constant parameter case is shown for comparison.}
    \label{fig:dubin car time-varying}
\end{figure}
\begin{figure}[ht]
   \centering
   \setlength{\belowcaptionskip}{-15pt}
   \hspace{-10pt}
   \includegraphics[width=0.5\columnwidth]{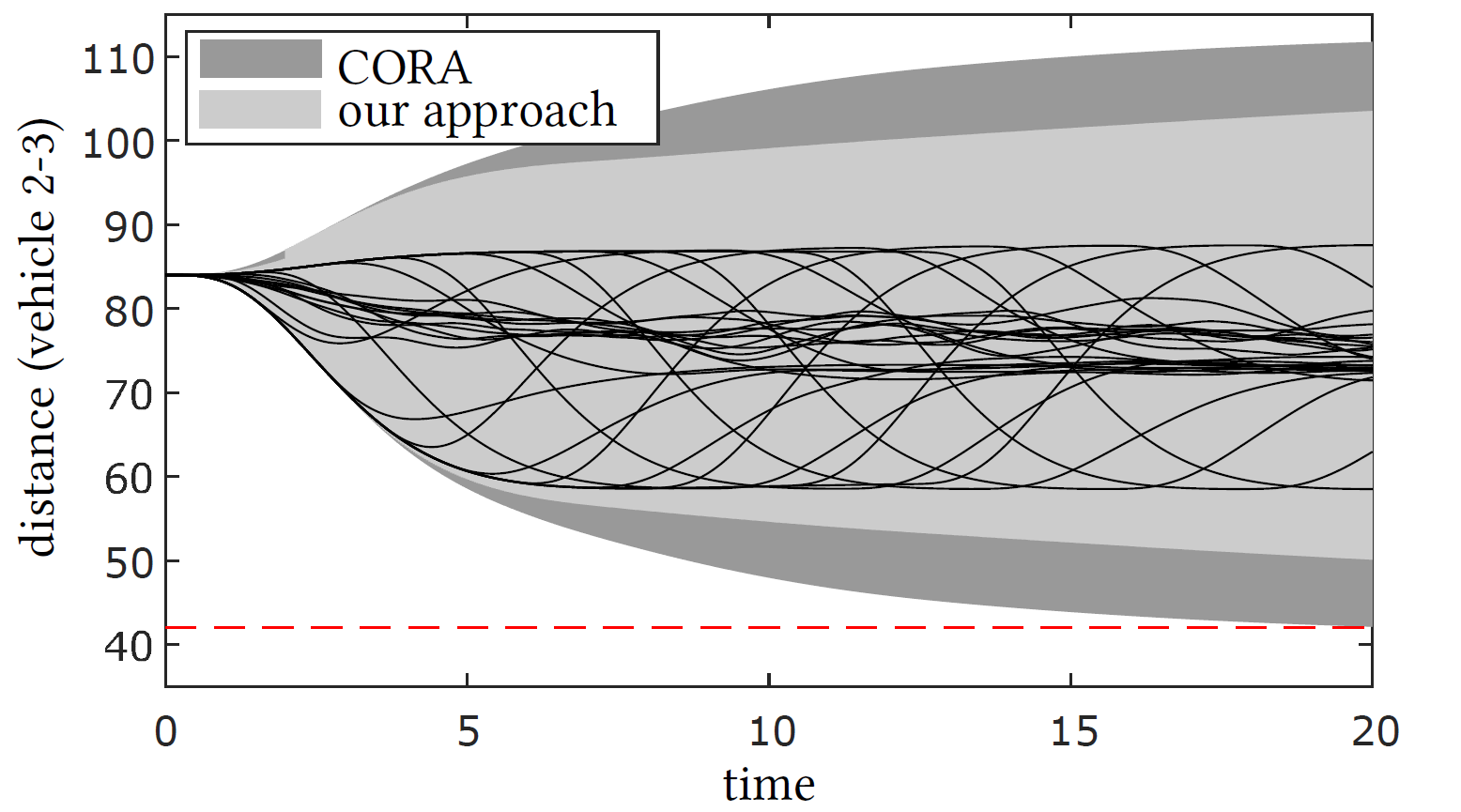}
   \vspace{-10pt}
    \caption{Comparison of the reachable set enclosure of the platoon benchmark computed with our approach and the CORA toolbox \cite{Althoff2015a}, where random simulations are shown in black and the required safe distance of $42$m is visualized in red.}
    \label{fig:Platoon}
\end{figure}
\vspace{\secminus}
\subsection{Dubins Car}

First, we consider the example of a Dubins car shown in Fig.~\ref{fig:dubinsCarSchematic} to demonstrate how well our reachabilty algorithm captures the non-convexity of the reachable set. According to \cite[Sec. 2.2]{Bak2022}, the dynamics of the Dubins car can be described as a four-dimensional linear parametric system. We consider a time horizon of $t_{end} = 2 \si{\second}$ together with the initial state $\mathcal{X}_0 = [0~0~0~10 \si{\meter \per \second}]$.
Fig. \ref{fig:dubin car comparison time interval} compares the reachable set enclosure for case with constant parameters computed with our algorithm and the zonotope method \cite{Althoff2011b}, which is implemented in the reachablity toolbox CORA \cite{Althoff2015a}. Our polynomial zonotope approach tightly captures the non-convexity of the reachable set, while the zonotope method yields a large convex over-approximation. Moreover, for our approach we can compute the reachable set using a single time step together with $\kappa = 60$ Taylor terms and a desired zonotope order of $\rho_d = 60$, which takes $0.25 \si{\second}$. For the zonotope approach, we obtained the tightest results for a time step size of $\Delta t = 0.05 \si{\second}$, $\kappa = 6$ Taylor terms, and desired zonotope order $\rho_d = 100$, which takes $0.18 \si{\second}$ to compute.  


As a second experiment, we consider the scenario where the parameters are time-varying. Since reachability analysis is more challenging in this case, we reduce the time horizon to $t_{end} = 1 \si{\second}$. We compare the approach presented in Sec.~\ref{sec:time varying parameter} to the zonotope method \cite{Althoff2011a}, which again is implemented in the reachability toolbox CORA \cite{Althoff2015a}. For the zonotope method we obtained the tightest result with the settings $\Delta t = 0.05 \si{\second}$, $\kappa = 4$, and $\rho_d = 30$, for which the calculation of the reachable set takes $0.069 \si{\second}$. For our approach, we use the same time step size but increase the other settings to $\kappa = 6$, and $\rho_d = 80$. While the calculation of the reachable set takes with $4 \si{\second}$ longer than for the zonotope method, our approach yields an enclosure that is in most regions more accurate as shown in Fig.~\ref{fig:dubin car time-varying}. Moreover, while the zonotope method only computes a very rough convex enclosure of the reachable set, our approach still captures the non-convexity of the reachable set to some extent. 

\vspace{\secminus}
\subsection{Vehicle Platoon}
To demonstrate the scalability of our approach, we now consider the PLAA01-BND42 instance of the $9$-dimensional platoon benchmark from the 2021 ARCH competition \cite[Sec.~3.7]{ARCH21:linear}, which describes a platoon consisting of four vehicles.
The verification task for the benchmark is to show that even in case of a communication loss the vehicles still keep a safe distance of 42m from each other.
This can be formulated as a reachability problem, where the corresponding parametric system encloses both the system with and without communication loss. 
The system has one uncertain input, which is the acceleration of the leading vehicle. 
For our experiments we slightly modify the benchmark by considering constant instead of time-varying parameters.
We compare our approach with the state of the art reachability toolbox CORA \cite{Althoff2015a}, which uses the zonotope method \cite{Althoff2011b} and requires a computation time of $15 \si{\second}$.
For our approach, calculation of the reachable set using Alg.~\ref{alg:mainReach} with $\Delta t = 0.036 \si{\second}$, $\kappa = 6$, and $\rho_d = 50$ takes $9 \si{\second}$, and is therefore faster than the CORA toolbox.
On top of that, as shown in Fig.~\ref{fig:Platoon} the reachable set enclosure computed with our approach is significantly tighter than the one computed with CORA.
In summary, our approach is therefore both faster and more accurate than the existing state of the art reachability tool.
\vspace{\secminus}
\section{Scalable Optimization for Multi-Affine Zonotopes}
\label{sec::scalable optimization}

\subsection{Reachability sets are multi-affine zonotopes}
Let us reconsider Algorithm \ref{alg:tvReach} from section 5.1. Let us make a (mild) assumption that the matrix sets   $\mathbfcal{A}$ and $\mathbfcal{B}$,    initial set $\mathcal{X}_0$,  input set $\mathcal{U}$ are indenpendent (true if they do not share a common uncertain factor).
Then, 
the resulting polynomial zonotope has a curiously useful property; the exponents over all the factors, minus the time factor, are at most 1.
This means that if the time factor is replaced with numerical values (as they are, after all, known), then \emph{the resulting reachable set is a multiaffine zonotope (Def \ref{def:multiaffine-zonotope}).} 
Intuitively, this is because we call $\textproc{fresh}$ operator on the power of $\mathbfcal{A}$, 
and because the particular solution is added through Minkowski sum for each time interval. 

Since multi-affine zonotopes are a type of  polynomial zonotopes, it is of course possible to apply existing set operations for polynomial zonotopes; however, usually more efficient methods exist when restricted to this special class.
Theoretically, as shown in \cite{huang2023difficulty}, optimization and intersection checking of polynomial zonotopes is NP-complete. This is also true of plotting polynomial zonotopes of dimension 2.
This is because both tasks rely on a 
splitting algorithm \cite{huang2023difficulty,Bak2022} 
which models polynomial zonotopes as a
union of zonotopes, created by sequentially splitting the original nonconvex set.
Given the importance of this splitting method, we now investigate how its efficiency can be improved over multiaffine zonotopes.

\begin{figure}
   \centering
   \setlength{\belowcaptionskip}{-15pt}
   \hspace{-8pt}
   \includegraphics[width=0.5\columnwidth]{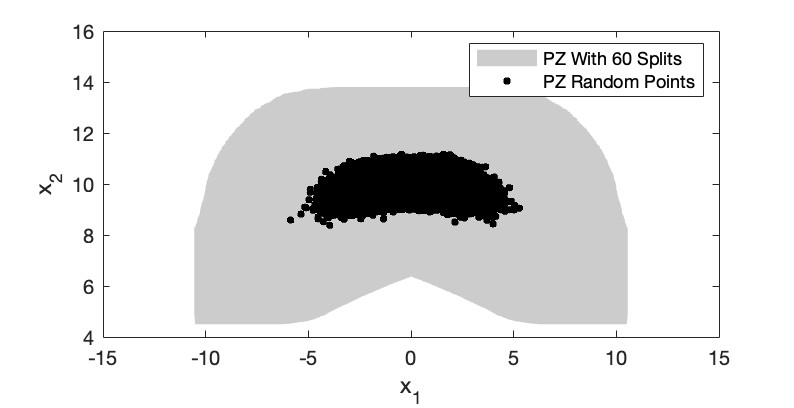}
   \vspace{-10pt}
    \caption{Time point reachable set of the Dubins car with time-varying parameters plotted through the default polynomial zonotope approach with 60 splits costing around 2 hours.}
    \label{fig:multi-affine motivation}
\end{figure}

\subsection{Splitting methods for polynomial zonotopes}
The core idea behind splitting methods for polynomial zonotopes is the observation that any nonconvex set can be represented (or overapproximated) as a union of zonotopes. This is a useful representation, since tasks such as intersection checking and 2-D plotting are much easier for zonotopes (can be done through quick index search) than polynomial zonotopes (requires nonconvex optimization). However, a drawback is that, even with careful implementation, the number of zonotopes may increase significantly, and as mentioned in \citep{huang2023difficulty}, it computationally demanding to achieve an accurate overrepresentation; in some implementations, it may even lead to an unbounded number of sets by not monotonically decreasing the overrepresentation set size in each round. 

As we see in Fig.\ref{fig:multi-affine motivation},  applying the built-in plotting algorithm for polynomial zonotope relying on splitting to plot the same time instance set as shown in Fig. \ref{fig:dubin car time-varying} takes around 2 hours, and may still result in a bad over-approximation.
One idea to improve the plotting process is to model the convex hull of the set as an intersection of hyperspaces by finding the supporting hyperplanes at various points on the boundary of the set.
These supporting hyperplanes can be derived by solving multi-affine optimization problems.

\subsection{Improved splitting methods for multiaffine zonotopes}
\label{sec:improve split for mtaff}
In this section we propose a scalable optimization algorithm for minimizing linear functions over multiaffine zonotopes. Let us start by defining the multi-affine optimization problem
\begin{definition}\label{multi-affine-def}
    Consider the polynomial defined as:
    \begin{align}\label{MultiAffine}
     p(\alpha_1,\alpha_2,\cdots,\alpha_n) = \sum_{I \subseteq \{1,2,\cdots,n \}} g_I \prod_{i \in I} \alpha_i,   
    \end{align}
    where $a_I \in \mathbb{R}$. The multi-affine polynomial optimization problem is:
    \begin{align}
        \min_{\substack{\alpha_1,...,\alpha_n \\ \alpha_i \in [-1,1] }} p(\alpha_1,\alpha_2,\cdots,\alpha_n).
        \label{eq:min_multiaffine}
    \end{align}
\end{definition}
Given a set \(C \subseteq \{1,2,\ldots, n\}\) and a multi-affine polynomial \(p\), as defined in Equation \eqref{MultiAffine}, we formalize the restriction of \(p\) to the set \(C\), denoted \(p|_{C}\), with the following expression:
\begin{equation}\label{eq:MultiaffineRestrict}
p|_{C} = \sum_{I \subseteq C} g_I \prod_{i \in I} \alpha_i
\end{equation}

Note that we can use multi-affine optimization to find a tight over-approximation of multi-affine zonotope. Suppose we have a multi-affine zonotope $\mathcal{MAZ}$. Let us start with a vector $\textbf{d}_1$. To find the support, we minimize the following:
\begin{equation}
\textbf{x}_{1} = \arg \min_{ \textbf{x} \in \mathcal{MAZ}} \textbf{d}_1^T \textbf{x}
\label{eq:proj_mtaff}
\end{equation}
Therefore, the half-space $\mathcal{H}_1$ outlined below include the multi-affine zonotope, and the associated supporting hyperplane, denoted as $\partial \mathcal{H}_1$, will be
\[
\mathcal{H}_1 = \{ \textbf{x} ~  | ~  \textbf{d}_1^{T} (\textbf{x}_1 - \textbf{x}) \leq 0 \}, \quad \partial \mathcal{H}_1 = \{ \textbf{x} ~  | ~ \textbf{d}_1^{T} (\textbf{x}_1 - \textbf{x}) = 0 \}
\]
Subsequently we can calculate $\mathcal{H}_2$. This iterative process will continue until we calculate $\mathcal{H}_m$. The specific procedure of how to choose $\textbf{d}_1,\cdots,\textbf{d}_m$ is so called the Kamenev method and being discussed in \cite{lotov2004interactive}. Hence this can be formalized to some incremental approach to find the directions to perform optimization, until the consecutive optimal solution values are close to each other under some small tolerance. As a result, $\mathcal{H}_1 \bigcap \mathcal{H}_2 \cdots \bigcap \mathcal{H}_m$ will be an over-approximation of the multi-affine zonotope. This type of convex over-approximation that can be combined with other methods, like  polynomial zonotopes splitting. This combination can make the single polynomial zonotope split method more effective, as we show in Section \ref{sec:scalable evaluation}. By merging these approaches, we can get a tighter approximation than using the method based on polynomial zonotope splits alone. A visualization of the above procedure is in figure \ref{fig:countor}.
\begin{figure}
   \centering
   \setlength{\belowcaptionskip}{-15pt}
   \hspace{-8pt}
   \includegraphics[width=0.5\columnwidth]{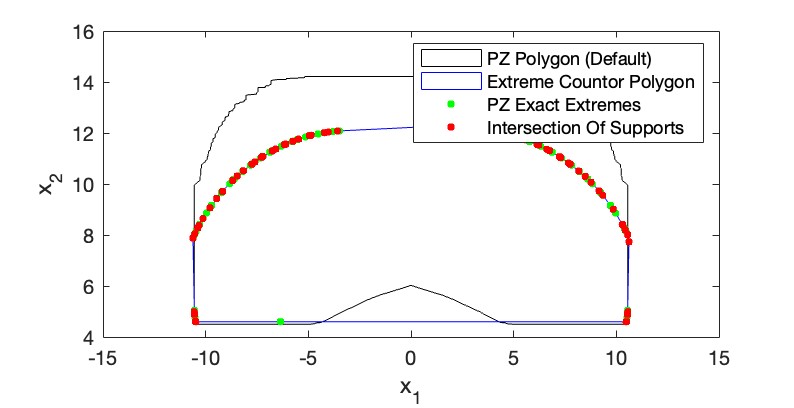}
   \vspace{-10pt}
    \caption{In the diagram, we initiate with a predefined series of directions identified as $d_1$. Subsequently, we identify a sequence of supports $\partial \mathcal{H}_1, \partial \mathcal{H}_2, \ldots, \partial \mathcal{H}_m$ and the respective points of intersection $\textbf{x}_1, \textbf{x}_2, \ldots, \textbf{x}_m$ with the multi-affine zonotope, label as green dots. The red dots are used to mark the points where adjacent supports intersect, specifically at $\partial \mathcal{H}_1 \bigcap \partial \mathcal{H}_2, \partial \mathcal{H}_2 \bigcap \partial \mathcal{H}3, \ldots, \partial \mathcal{H}_{m-1} \bigcap \partial \mathcal{H}_m$.}
    \label{fig:countor}
\end{figure}

We now present our method for solving \eqref{eq:min_multiaffine}.
Since all variables in a multi-affine optimization problem have an exponent of one, the partial derivative along each variable cannot change sign.
This means that the optimal value must occur on one of the corners of the $n$-dimensional box of the domain, and it is sufficient to consider this finite set when optimizing.
\begin{align*}
        \min_{\substack{\alpha_1,...,\alpha_n \\\alpha_i \in [-1,1] }} p(\alpha_1,\alpha_2,\cdots,\alpha_n) = \min_{\substack{\alpha_1,...,\alpha_n \\\alpha_i \in \{-1\} \cup \{1\}}} p(\alpha_1,\alpha_2,\cdots,\alpha_n).
\end{align*}
The above problem is, in general, NP-complete due to the following theorem:
\begin{theorem} \citep{huang2023difficulty}
Optimization over bilinear polynomial optimization is NP-complete where bilinear polynomial is defined as 
\begin{align}
\label{eq:bilinear_opt}
    \min_{\substack{\alpha_1,...,\alpha_n \\ \alpha_i \in \{-1\} \cup \{1\} }} \sum_{i=1}^{n} \sum_{j=i+1}^{n} g_{i,j} \alpha_i \alpha_j.
\end{align}
\end{theorem}
We can design a brute force algorithm to find the optimal value of the multi-affine polynomial. 
\begin{algorithm}[H]
\caption{Brute Force Minimization of Multi-Affine Polynomial}
\label{alg:non-recursive-brute-force}
\begin{algorithmic}[1]
\Function{Minimize}{$p$, $n$} 
\State $m \gets \infty$
\For{$i \gets 0$ \textbf{to} $2^n-1$} 
    \For{$j \gets 1$ \textbf{to} $n$}
    \State $\alpha_j = 2 \left( \left\lfloor \frac{i}{2^j} \right\rfloor \mod 2 \right) - 1$
    \EndFor
    \State $m \gets \min(m,p(\alpha_1, \alpha_2, \ldots, \alpha_n))$
\EndFor
\State \Return $m$
\EndFunction
\end{algorithmic}
\end{algorithm}

The algorithm, as defined, has exponential running time, aligning with the NP-Completeness of the problem at hand. In the worst case, 
it may check all  $2^n$ possible values for $\{\alpha_i\}_i^n$.
However, it is possible that some problems can be solved more efficiently.
Specifically, the work in \citep{del2023complexity} highlights that the complexity of the multi-affine polynomial optimization problem is closely linked to the associated dependency graph structure. 
\begin{definition}
Given a multiaffine polynomial \( p \) as defined in \eqref{MultiAffine}, the associated \textbf{factor dependency graph} \( G = (V, E) \) is constructed as follows:
\begin{itemize}
    \item The vertex set \( V = \{1,...,n\}\) represents  the indices of the variables \( \alpha_1, \alpha_2, \ldots, \alpha_n \) in the polynomial \( p \).
    \item The edge set \( E \subset V\times V\) contains an edge \( e = \{i, j\} \in E\) 
    if there exists a monomial in the polynomial \( p \) \eqref{MultiAffine} that includes both coefficients $a_i$ and $a_j$.
\end{itemize}
\end{definition}
This graph represents the pairwise dependencies between variables in the polynomial optimization problem, where an edge signifies that the connected variables jointly affect the value of the polynomial.

Suppose we have a multi-affine polynomial \( p \) and its associated dependency graph \( G = (V, E) \). If \( G \) is disconnected, e.g. consisting of disconnected   components \( G_1 = (V_1, E_1), G_2 = (V_2, E_2), \ldots, G_m = (V_m, E_m) \), (e.g. where there is no edge $\{i,j\}\in E$ where $i\in V_k$, $j\in V_l$, $k\neq l$), then the optimization problem can be decomposed over each 
connected component:
\[
\min_{\alpha_i \in [-1,1],i \in V} p = \min_{\alpha_i \in [-1,1],i \in V_1} p_1 + \min_{\alpha_i \in [-1,1],i \in V_2} p_2 + \cdots + \min_{\alpha_i \in [-1,1],i \in V_m} p_m
\]
Here, $p_i$ only contains the monomials from $p$ that involve coefficients represented by nodes in $V_i$.
Then, the runtime reduces from $2^n$ to $\sum_i 2^{|V_i|}$, which is often significant in practice.
Let's illustrate this with the following example:

\begin{example}\label{example:multi}
Consider a multi-affine polynomial \( p \) and its associated dependency graph \( G = (V, E) \) defined as follows:
\begin{align*}
    p &= g_1 \alpha_1 \alpha_2 \alpha_3 + g_2 \alpha_3 \alpha_4 + g_3 \alpha_5 \alpha_6 + g_4 \alpha_7, \\
    V &= \{\alpha_1, \alpha_2, \alpha_3, \alpha_4, \alpha_5, \alpha_6, \alpha_7\}, \\
    E &= \{\{\alpha_1, \alpha_2\}, \{\alpha_1, \alpha_3\}, \{\alpha_2, \alpha_3\}, \{\alpha_3, \alpha_4\}, \{\alpha_5, \alpha_6\}\}. 
\end{align*}

The dependency graph \( G \) can be divided into the following connected components:
\begin{itemize}
    \item \( G_1 = (V_1, E_1) \) with \( V_1 = \{\alpha_1, \alpha_2, \alpha_3, \alpha_4\} \) and \( E_1 = \{\{\alpha_1, \alpha_2\}, \{\alpha_1, \alpha_3\}, \{\alpha_2, \alpha_3\}, \{\alpha_3, \alpha_4\}\} \),
    \item \( G_2 = (V_2, E_2) \) with \( V_2 = \{\alpha_5, \alpha_6\} \) and \( E_2 = \{\{\alpha_5, \alpha_6\}\} \),
    \item \( G_3 = (V_3, E_3) \) consisting of the isolated vertex \( V_3 = \{\alpha_7\} \) with no edges.
\end{itemize}

The polynomial \( p = p_1 + p_2 + p_3\) can be decomposed into sub-polynomials \( p_1 \), \( p_2 \), and \( p_3 \) corresponding to these components:
\begin{align} \label{example:dependency}
    p_1 &= g_1 \alpha_1 \alpha_2 \alpha_3 + a_2 \alpha_3 \alpha_4, \\
    p_2 &= g_3 \alpha_5 \alpha_6, \\
    p_3 &= g_4 \alpha_7.
\end{align}

Thus, the original optimization problem can be decomposed into:
\begin{align*}
    \min_{\alpha_i \in [-1,1], i \in V} p = \min_{\alpha_1, \alpha_2, \alpha_3, \alpha_4 \in [-1,1]} p_1 + \min_{\alpha_5, \alpha_6 \in [-1,1]} p_2 + \min_{\alpha_7 \in [-1,1]} p_3.
\end{align*}
Without decomposition, Alg. \ref{alg:non-recursive-brute-force}
would enumerate \( 2^7 = 128 \) possibilities; however, through decomposition,
we only need to enumerate \( 2^4 + 2^2 + 2 = 22 \) possibilities, which is a significant reduction.
\end{example}
Now suppose an associated graph is \emph{almost} disconnected; for example, there is only one vertex, that, if removed, would disconnect the graph.
Then we propose to remove this vertex so that the dependency graph can be partitioned into multiple connected components. Let us illustrate it by considering the following example:

\begin{example}\label{example:removal}
Consider the multi-affine polynomial optimization problem defined by the polynomial:
\[
p(\alpha_1, \alpha_2, \alpha_3, \alpha_4, \alpha_5, \alpha_6, \alpha_7,\alpha_8) = (g_1 \alpha_1 \alpha_2 \alpha_3 + a_2 \alpha_3 \alpha_4
+ g_3 \alpha_5 \alpha_6, + g_4 \alpha_7) \alpha_8.
\]
where \( g_1, g_2, g_3, \) and \( g_4 \) are real coefficients. The associated dependency graph \( G = (V, E) \) for this polynomial is given by:
\begin{itemize}
    \item $ V = \{\alpha_1, \alpha_2, \alpha_3, \alpha_4, \alpha_5, \alpha_6, \alpha_7,\alpha_8\} $
    \item $E = \{\{\alpha_1, \alpha_2\}, \{\alpha_1, \alpha_3\}, \{\alpha_2, \alpha_3\}, \{\alpha_3, \alpha_4\}, \{\alpha_5, \alpha_6\}\{\alpha_1, \alpha_8\},$ \\
    $ \{\alpha_2, \alpha_8\},\{\alpha_3, \alpha_8\},\{\alpha_4, \alpha_8\},\{\alpha_5, \alpha_8\},\{\alpha_6, \alpha_8\},\{\alpha_7, \alpha_8\}\}$ 
\end{itemize}

In this dependency graph, vertex \( \alpha_8 \) is a key node connecting multiple edges. The removal of \( \alpha_8 \) leads to partitioning the graph into separate connected components. The remaining graph, after removing \( \alpha_8 \), can be separated into three distinct connected components


\[
V_1 = \{\alpha_1, \alpha_2, \alpha_3,\alpha_4\}, \quad V_2 = \{\alpha_5, \alpha_6\} \quad V_3 = \{\alpha_7\}.
\]
More specifically, if $\alpha_8=1$
\begin{align*}
    \min_{\alpha_i \in [-1,1], i \in V} p_+ = \min_{\alpha_1, \alpha_2, \alpha_3, \alpha_4 \in [-1,1]} p_1 + \min_{\alpha_5, \alpha_6 \in [-1,1]} p_2 + \min_{\alpha_7 \in [-1,1]} p_3.
\end{align*}
where 
\[
p_1 = g_1 \alpha_1 \alpha_2 \alpha_3 +  g_2 \alpha_3 \alpha_4, \qquad p_2 =  g_3 \alpha_5 \alpha_6, \qquad p_3 = g_4 \alpha_7.
\]
Similarly, if $\alpha_8=-1$, we simply flip signs where needed, and minimize
\begin{align*}
    \min_{\alpha_i \in [-1,1], i \in V} p_- = \min_{\alpha_1, \alpha_2, \alpha_3, \alpha_4 \in [-1,1]} (-p_1)  +\min_{\alpha_5, \alpha_6 \in [-1,1]} (-p_2) + \min_{\alpha_7 \in [-1,1]} (-p_3).
\end{align*}
Then the minimum of $p$ is $\min\{p_+,p_-\}$.
\end{example}
\begin{figure}[ht]
    \centering
    \begin{subfigure}[b]{0.45\linewidth}        \includegraphics[width=\linewidth]{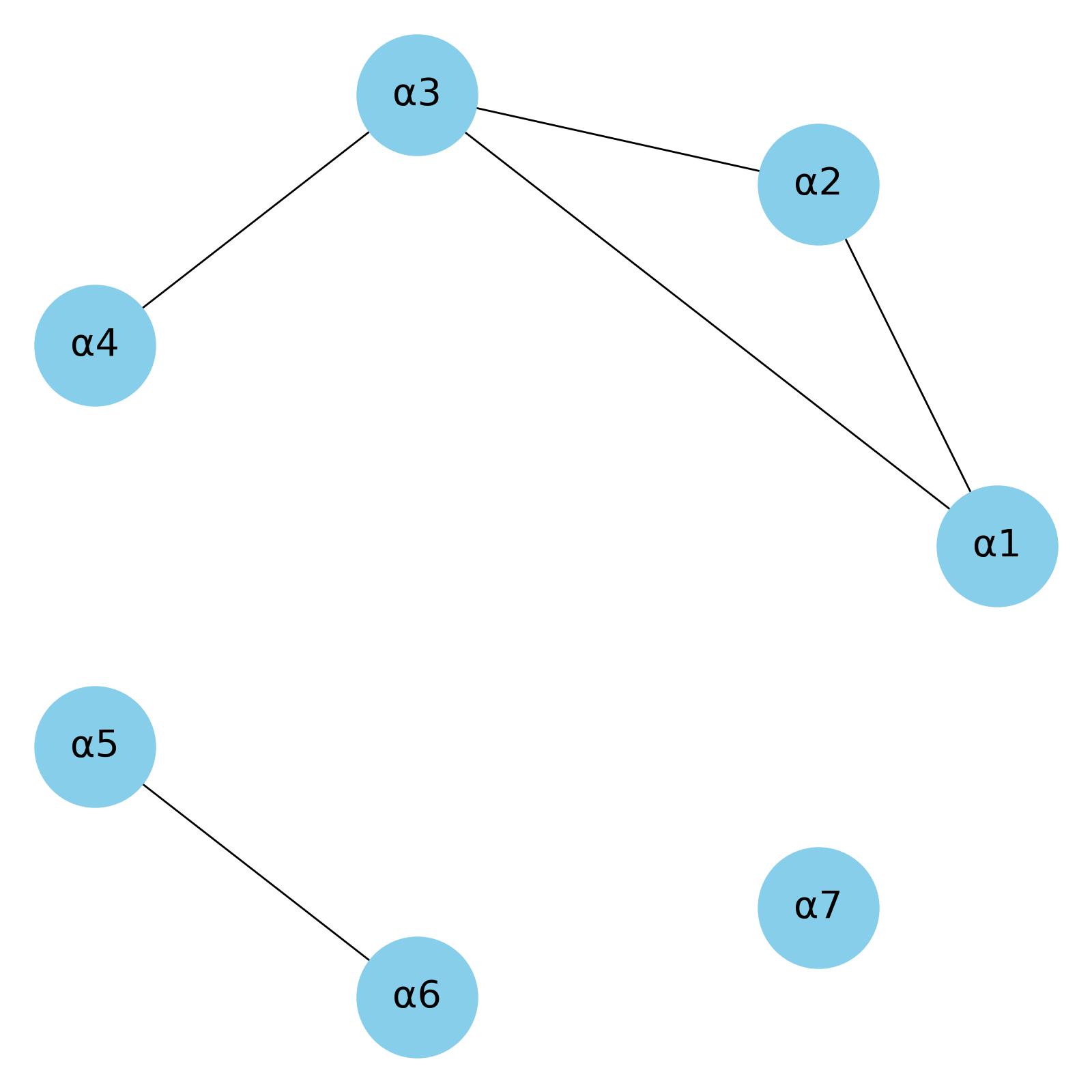}
        \caption{ Factor dependency graph for Example \ref{example:multi}}
        \label{fig:sub1}
    \end{subfigure}
    \hfill
    \begin{subfigure}[b]{0.45\linewidth}
    \includegraphics[width=\linewidth]{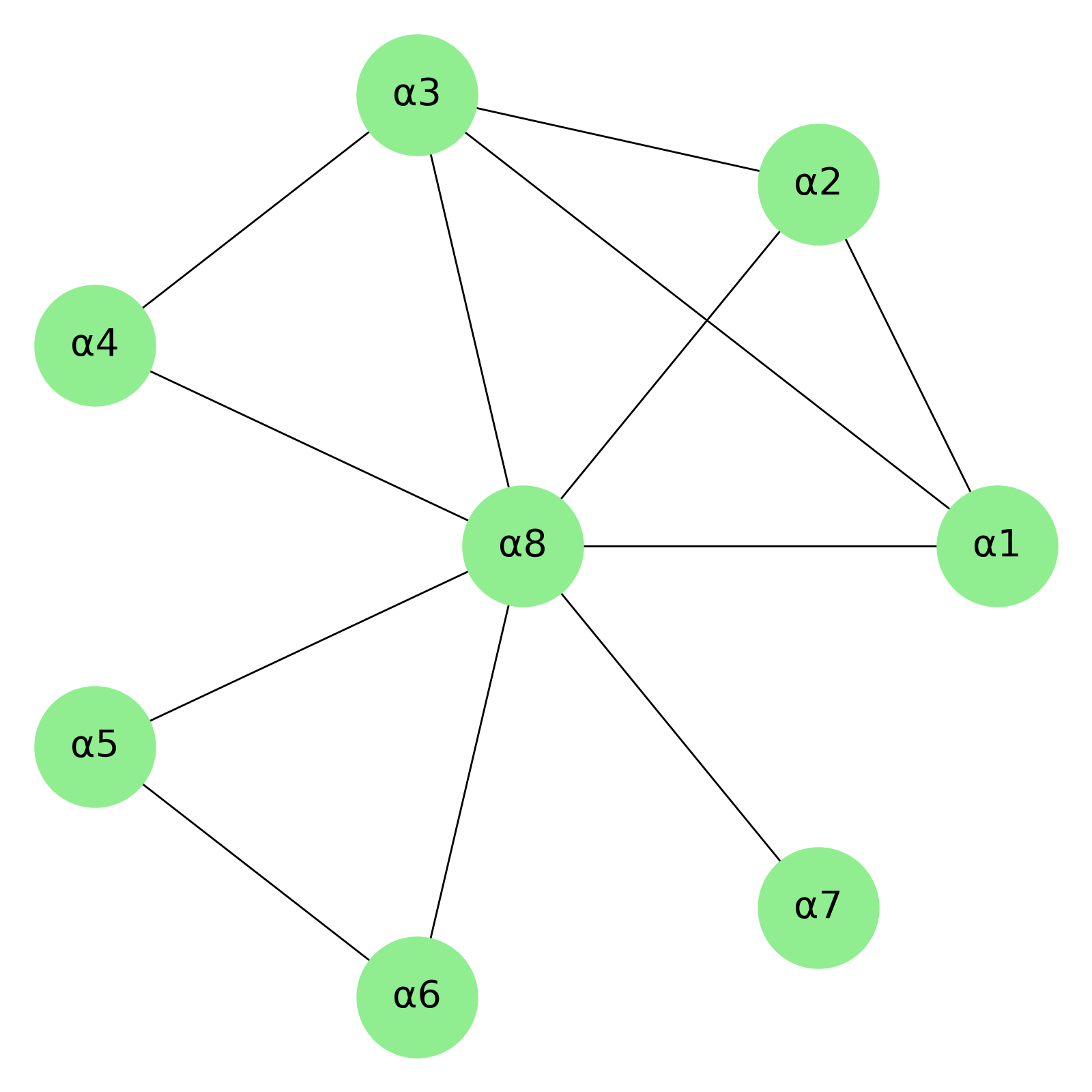}
        \caption{Factor dependency graph for Example \ref{example:removal}}
        \label{fig:sub2}
    \end{subfigure}
    \caption{The visualization shows the dependency graph. By removing the node $\alpha_8$ from the graph on the right, we see that the graph is split into three separate parts, as shown on the left.}
    \label{fig:graphs}
\end{figure}
Given the above example, we design the Algorithm \ref{alg:multi} and \ref{alg:multione} to solve the multi-affine optimization problem efficiently.

\begin{algorithm}
\caption{Minimize a Multi-affine Polynomial Based on Splitting}
\label{alg:multi} 
\begin{algorithmic}[1]
\Function{SplitMin}{$p$, $g$, $t$}
    \State \textbf{Input}: A multi-affine polynomial $p$, associated dependency graph $g = (v, e)$, threshold $t$.
    \State \textbf{Output}: Optimal value of polynomial $p$.
    
    \If{$|v| \leq t$}
        \State \Return $\Call{Minimize}{p}$ 
    \EndIf
    
    \State $m \gets 0$
    \State $gc \gets \Call{ConnectedComponents}{g}$ 
    
    \ForAll{$c \in gc$}
        \State $m \gets m + \Call{MinOneComponent}{p, g, c, t}$
    \EndFor
    
    \State \Return $m$
\EndFunction
\end{algorithmic}
\end{algorithm}

\begin{algorithm}
\caption{Function for Minimizing Over a Single Component}
\label{alg:multione}
\begin{algorithmic}[1]
\Function{MinOneComponent}{$p$, $g$, $c$, $t$}
 \State \textbf{Input}: A multi-affine polynomial $p$, associated dependency graph $g = (v, e)$, threshold $t$ and
 a connected component $c$ from graph $g$. 
\State \textbf{Output}: Optimal value of polynomial $p|_c$ defined in \eqref{eq:MultiaffineRestrict}
    \If{$|c| \leq t$}
        \State \Return $\Call{Minimize}{p|_c}$
    \Else
        \State $ct \gets \Call{FindMinVertexCuts}{g, c}$\footnotemark
        \Comment{Cut on the connected component c from g}
        \State Initialize a stack $s$ with $p|_c$ at the top
        \ForAll{$nd \in ct$}
        \Comment{Each vertex return by the cut}
            \State New stack $s'$ is empty
            \While{$s$ is not empty}
                \State $p' \gets s.pop()$
                \For{$i \in \{-1,1\}$}
                    \State $p' \gets \Call{Substitute}{p', nd, i}$ 
                    \State $s'.push(p')$
                \EndFor
            \EndWhile
            \State $s \gets s'$ 
        \EndFor
        \State $m \gets \infty$
        \While{$s$ is not empty}
            \State $p' \gets s.pop()$
            \State $g' \gets \Call{RemoveVertices}{g, ct}$
            \State $m \gets \min(m, \Call{SplitMin}{p', g', t})$
        \EndWhile
        \State \Return $m$
    \EndIf
\EndFunction
\end{algorithmic}
\end{algorithm}
\footnotetext{The algorithm returns the minimum vertex cut \cite{west2001introduction}  of the connected component $c$.}
\subsection{Evaluation}
\label{sec:scalable evaluation}
In this section, we perform two experiments to demonstrate how the multi-affine zonotope specialized optimization algorithm \ref{alg:multi} and \ref{alg:multione} can help relax the difficulties existing in several crucial operations, which are based on the splitting stragety \cite{Bak2022}, of polynomial zonotopes, including 2-d plotting, support function computation and intersection checking. We compare our method with the newest CORA v2024.1.3, which is a state-of-the-art polynomial zonotopes tool box. Our optimization method is implemented in Python 3.9 and experiments run on an Apple M2 Pro processor with 16GB memory.
\footnote{In contrast, the reachability evaluation in section \ref{sec:Application} compares against CORA toolbox which is the one at the time of the HSCC conference paper getting reviewed, for the purpose of consistency.} 

For the first experiment, we notice that in figure \ref{fig:multi-affine motivation} the 2-d reachable set projection plotted by the CORA built-in splitting based  approach for polynomial zonotope does not provide a tight over-approximation of the reachable set after running around 2 hours. 
Instead, as we notice that time-point reachable set is actually a multi-affine zonotope, we can apply our multi-affine optimization approach to provide a tighter over-approximation which is given in figure \ref{fig:multi-affine}.
The grey region corresponds to CORA polynomial zonotope plotting approach with 60 splits, which takes time more than 2 hours. The blue contour corresponds to getting the convex enclosure taking the idea from Sec. \ref{sec:improve split for mtaff} and apply through our multi-affine optimization approach. After that we intersect our convex contour with the non-convex over-approximated polygon obtained through the CORA plotting and splitting process with 20 splits. 
The computation time of getting the convex contour, including the Kamenev approach to find the direction to optimize with tolerance of $10^{-7}$, and the graph based multi-affine optimization approach itself, combined with the computation time of the CORA built-in function with 20 splits takes $1063s$ in total. Through this experiment, we demonstrate from a visualized perspective that the splitting based algorithm would cost huge computation time when the split number is high for the purpose of getting a tighter over-approximation enclosure. However, the reality is that though huge computation budget is spent, the gain is limited. 

In this second experiment, we are going to make the same argument about splitting based approaches, which have the imbalance of the accuracy improvement with the high computation cost, through numerical evaluation of the supports functions. To demonstrate the inadequacies of the splitting approach not only exist in the time point reachable set from the Dubin's car dynamic with time-varying parameters, but also is a more general problem. We experiment over a different 5-D benchmark system with time-varying parameters from section 5.1 of \citep{Althoff2011b}. We perform reachability analysis through Algorithm. \ref{alg:tvReach} by using a high zonotope order of 500. Then we get the third time point reachable set, and reduce it to a polynomial zonotope (which is multi-affine) with $495$ dependent generators, $5$ independent generators, $96$ dependent factors and project to the first two dimension as our experiment target. The projection would not prevent us from showing the scalability of our result since one always need to project the set representation onto some direction to perform optimization process (by eq. \ref{eq:proj_mtaff}). 

We run the incremental Kamenev approach to find 6 directions to optimize with tolerance of $10^{-2}$. Then test our optimization approach for each direction and comparing against the CORA splitting based support function computation approach for $10$, $20$, $30$ splits. The results are given in Table \ref{tab:all_cases}. Since multi-affine optimization always returns the optimal solution, we set the optimal values we get as a comparison standard $f^{*}$ for the CORA method. Let the optimal value given by CORA support function as $f$. We define the relative error as $\frac{|f - f^{*}|}{|f^{*}|}$. In our study, we found that when we increase the number of splits in the support function, which makes the computation time longer, the improvement in error reduction is not as significant as expected. This result shows a clear imbalance between the high computational cost from more splits and the small gains in accuracy. On the other hand, directly apply our multi-affine optimization approach, guarantees to find the optimal solution while having very few computation time.

\begin{figure}
   \centering
   \setlength{\belowcaptionskip}{-15pt}
   \hspace{-8pt}
   \includegraphics[width=0.5\columnwidth]{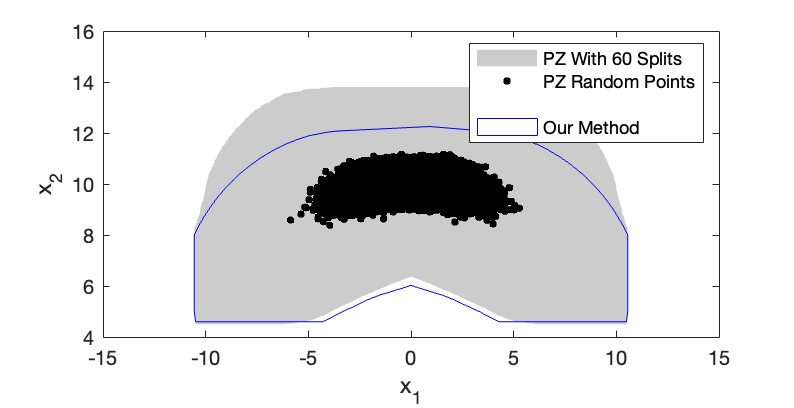}
   \vspace{-10pt}
    \caption{Time point reachable set of the Dubins car dynamics with time varying parameters. The grey region corresponds to CORA polynomial zonotope plotting approach with 60 splits, takes time more than 2 hours. The blue line corresponds to by using our multi-affine optimization approach with polynomial approach with 20 splits.}
    \label{fig:multi-affine}
\end{figure}
\begin{table}[htbp]
\centering
\caption{Relative error and Running Time for different methods.}
\label{tab:all_cases}
\renewcommand{\arraystretch}{1.3} 

\begin{tabular}{@{}llll@{}} 
\toprule
Method & Case 1 (Error/ Time) & Case 2 (Error/ Time) & Case 3 (Error/ Time) \\
\midrule
10 Splits & $0.41\%$ / \textbf{0}s & $0.6\%$ / \textbf{1}s & $0.49\%$ / \textbf{1}s \\
20 Splits & $0.36\%$ / 9s & $0.54\%$ / 54s & $0.49\%$ / 344s \\
30 Splits & $0.34\%$ / 227s & $0.52\%$ / 5818s & $0.49\%$ / 3087s \\
Ours & \textbf{0} / 16s & \textbf{0} / 16s & \textbf{0} / 15s \\
\bottomrule
\end{tabular}

\vspace{5mm} 

\begin{tabular}{@{}llll@{}}
\toprule
Method & Case 4 (Error / Time) & Case 5 (Error/ Time) & Case 6 (Error/ Time) \\
\midrule
10 Splits & $0.98\%$ / \textbf{0}s & $0.74\%$ / \textbf{1}s & $0.79\%$ / \textbf{0}s \\
20 Splits & $0.88\%$ / 43s & $0.72\%$ / 358s & $0.64\%$ / 330s \\
30 Splits & $0.83\%$ / 1156s & Exceeds limit & $0.59\%$ / 521s \\
Ours & \textbf{0} / 16s & \textbf{0} / 15s & \textbf{0} / 15s \\
\bottomrule
\end{tabular}

\end{table}

\section{Conclusion}
\label{sec:conclusion}
In this work, we used polynomial zonotopes---a set representation originally developed to analyze nonlinear systems---in order to examine four different flavors of linear systems.
Since polynomial zonotopes preserve dependencies among both symbolic variables and time, the method computes tight non-convex enclosures of reachable sets.
Comparisons on two benchmark systems illustrate that our reachability algorithm yields much tighter enclosures than other state-of-the-art methods.
For linear time invariant systems, our method can represent the reachable set for the entire time horizon with a single polynomial zonotope. 
We extend the conference paper in to this journal version by introducing an efficient and scalable algorithm to optimize over multi-affine zonotopes, which is a special case of polynomial zonotopes.
Through experiments we demonstrate the efficiency and scalability of our algorithm and reveal the problem of existing state-of-the-art method. However, it still faces some problems. For instance, if the way we split the data (vertex cut) is uneven, the algorithm takes longer to run. Creating an algorithm that always splits the data evenly (balanced cut) is also difficult. For future work, we want to explore the possibility of making a random algorithm that can split the data evenly, or we might try a different method that doesn't require checking every possible option for each factor.

\bibliographystyle{plainnat}
\bibliography{kochdumper,cpsGroup, newRef}

\vspace{20pt}




\end{document}